\documentclass[letterpaper, 11pt]{article}
\usepackage[margin=1in]{geometry}
\usepackage{latexsym}
\usepackage{amssymb}
\usepackage{times}
\usepackage{xcolor}
\usepackage{color}
\usepackage{amsmath,amsfonts,amsthm}
\usepackage{graphicx,float}
\usepackage{enumitem}
\usepackage[ruled,vlined,resetcount, algosection]{algorithm2e}
\usepackage{setspace}

\usepackage{hyperref}
\makeatletter
\def\namedlabel#1#2{\begingroup
    #2%
    \def\@currentlabel{#2}%
    \phantomsection\label{#1}\endgroup
}
\makeatother

\hypersetup{
bookmarks=true,
colorlinks=true,
linkcolor=blue,
urlcolor=blue,
citecolor=blue,
pdftex,
linktocpage=true, 
hyperindex=true
}

\newcommand{\poly}{\mathsf{poly}}

\newcommand{\AC}{\mathsf{AC}}

\newcommand{\eps}{\varepsilon}
\newcommand{\bit}{\{0, 1\}}

\newcommand{\N}{\mathsf{N}}

\newcommand{\ED}{\mathsf{ED}}
\newcommand{\LCS}{\mathsf{LCS}}

\newtheorem{lemma}{Lemma}[section]
\newtheorem{theorem}[lemma]{Theorem}

\newtheorem{construction}[lemma]{Construction}

\newtheorem{definition}[lemma]{Definition}
\newtheorem{corollary}[lemma]{Corollary}

\newtheorem{remark}[lemma]{Remark}

\begin{document}

\allowdisplaybreaks

\begin{titlepage}
\def\thepage{}

\title{Deterministic Document Exchange Protocols, and Almost Optimal Binary Codes for Edit Errors}
\author{Kuan Cheng  \thanks{kcheng17@jhu.edu.\ Department of Computer Science, Johns Hopkins University. Supported by NSF Grant CCF-1617713.} \and Zhengzhong Jin  \thanks{zjin12@jhu.edu.\ Department of Computer Science, Johns Hopkins University.\ Partially supported by NSF Grant CCF-1617713.}\and Xin Li  \thanks{lixints@cs.jhu.edu.\ Department of Computer Science, Johns Hopkins University. Supported by NSF Grant CCF-1617713.} \and Ke Wu \thanks{AshleyMo@jhu.edu.\ Department of Computer Science, Johns Hopkins University.} }

\maketitle \thispagestyle{empty}


\begin{abstract}
We study two basic problems regarding edit errors (insertions and deletions). The first one is document exchange, where two parties Alice and Bob hold two strings $x$ and $y$ with a bounded edit distance $k$. The goal is to have Alice send a short sketch to Bob, so that Bob can recover $x$ based on $y$ and the sketch. The second one is the fundamental problem of designing error correcting codes for edit errors, where the goal is to construct an explicit code to transmit a message $x$ through a channel that can add at most $k$ worst case insertions and deletions, so that the original message $x$ can be successfully recovered at the other end of the channel. Both problems have been extensively studied for decades, and in this paper we focus on deterministic document exchange protocols and binary codes for insertions and deletions (insdel codes). If the length of $x$ is $n$, then it is known that for small $k$ (e.g., $k \leq n/4$), in both problems the optimal sketch size or the optimal number of redundant bits is $\Theta(k \log \frac{n}{k})$. In particular, this implies the existence of binary codes that can correct $\eps$ fraction of insertions and deletions with rate $1-\Theta(\eps \log (\frac{1}{\eps}))$. However, known constructions are far from achieving these bounds.

In this paper we significantly improve previous results on both problems. For document exchange, we give an efficient deterministic protocol with sketch size $O(k \log^2 \frac{n}{k})$. This significantly improves the previous best known deterministic protocol, which has sketch size $O(k^2 + k \log^2 n)$ \cite{Belazzougui2015EfficientDS}. For binary insdel codes, we obtain the following results:

\begin{enumerate}
\item An explicit binary insdel code which encodes an $n$-bit message $x$ against $k$ errors with redundancy $O(k \log^2 \frac{n}{k})$. In particular this implies an explicit family of  binary insdel codes that can correct $\eps$ fraction of insertions and deletions with rate $1-O(\eps \log^2 (\frac{1}{\eps}))=1-\widetilde{O}(\eps)$. This significantly improves the previous best known result which only achieves rate $1-\widetilde{O}(\sqrt{\eps})$ \cite{7835185}, \cite{7541373}, and is optimal up to a $\log (\frac{1}{\eps})$ factor.

\item An explicit binary insdel code which encodes an $n$-bit message $x$ against $k$ errors with redundancy $O(k \log n)$. This significantly improves the previous best known result of \cite{8022906}, which only works for constant $k$ and has redundancy $O(k^2 \log k \log n)$; and that of \cite{Belazzougui2015EfficientDS}, which has redundancy $O(k^2 + k \log^2 n)$. Our code has optimal redundancy for $k \leq n^{1-\alpha}$, any constant $0<\alpha<1$. This is the first explicit construction of binary insdel codes that has optimal redundancy for a wide range of error parameters $k$, and this brings our understanding of binary insdel codes much closer to that of standard binary error correcting codes.
\end{enumerate}

In obtaining our results we introduce several new techniques. Most notably, we introduce the notion of \emph{$\eps$-self matching hash functions} and \emph{$\eps$-synchronization hash functions}. We believe our techniques can have further applications in the literature.
\end{abstract}

\end{titlepage}

\section{Introduction}
Given two strings $x, y$ over some finite alphabet $\Sigma$, the edit distance between them $\ED(x, y)$ is defined as the minimum number of  edit operations (insertions, deletions and substitutions) to change $x$ into $y$. Being one of the simplest metrics, edit distance has been extensively studied due to its wide applications in different areas.  For example, in natural language processing, edit distance is used in automatic spelling correction to determine possible corrections of a misspelled word; and in bioinformatics it can be used to measure the similarity between DNA sequences.\ In this paper, we study the general question of recovering from errors caused by edit operations. Note that without loss of generality we can only consider insertions and deletions, since a substitution can be replace by a deletion followed by an insertion, and this at most doubles the number of operations. Thus from now on we will only be interested in insertions and deletions, and we define $\ED(x, y)$ to be the minimum number of such operations required to change $x$ into $y$.

Insertion and deletion errors happen frequently in practice. For example, they occur in the process of reading magnetic and optical media, in genetic mutation where DNA sequences may change, and in internet protocols where some packets may get lost during routing. Another typical situation where these errors can occur is in distributed file systems,  e.g.,  when a file is stored in different machines and being edited by different people working on the same project. These files then may have different versions that need to be synchronized to remove the edit errors. In this context, we study the following two basic problems regarding insertion and deletion errors.

\begin{itemize}
\item \emph{Document exchange.} In this setting, two parties Alice and Bob each holds a string $x$ and $y$, and we assume that their edit distance is bounded by some parameter $k$. The goal is to have Alice send a sketch to Bob based on her string $x$ and the edit distance bound $k$, such that Bob can recover Alice's string $x$ based on his string $y$ and the sketch. Naturally, we would like to require both the message length and the computation time of Alice and Bob to be as small as possible.

\item \emph{Error correcting codes.} In this setting, two parties Alice and Bob are linked by a channel where the number of worst case insertion and deletions is bounded by some parameter $k$. Given any message, the goal is to have Alice send an encoding of the message to Bob through the channel, so that despite any possible insertion and deletion errors that may happen, Bob can recover the correct message after receiving the (possibly modified) codeword. Again, we would like to minimize both the codeword length (or equivalently, the number of redundant bits) and the encoding/decoding time. This is a generalization of the classical error correcting codes for Hamming errors.
\end{itemize}

It can be seen that these two problems are closely related. In particular, a solution to the document exchange problem can often be used to construct an error correcting code for insertion and deletion errors. In this paper, we focus on the setting where the strings have a binary alphabet, arguably the most popular and important setting in computer science.\footnote{Although, our document exchange protocols can be easily extended to larger alphabets, we omit the details here.}\ In this case, assume that Alice's string (or the message she wants to send) has length $n$, then it is known that for small $k$ (e.g., $k \leq n/4$) both the optimal sketch size in document exchange and the optimal number of redundant bits in an error correcting code is $\Theta(k \log (\frac{n}{k}))$, and this is true even for Hamming errors. In addition, both optimum can be achieved using exponential time, with the first one using a greedy coloring algorithm and the second one using a greedy sphere packing algorithm (which is essentially what gives the Gilbert-Varshamov bound).

It turns out that in the case of Hamming errors, both optimum (up to constants) can also be achieved efficiently in polynomial time.\ This is done by using sophisticated linear Algebraic Geometric codes \cite{hoholdt1998algebraic}. As a special case, one can use Reed-Solomon codes to achieve $O(k \log n)$ in both problems. However, the situation becomes much harder once we switch to edit errors, and our understanding of these two basic problems lags far behind the case of Hamming errors.\ We now survey related previous work below.

\paragraph{Document exchange.} Historically, Orlitsky \cite{185373} was the first one to study the document exchange problem. His work gave protocols for generally correlated strings $x, y$ using a greedy graph coloring algorithm, and in particular he obtained a deterministic protocol with sketch size $O(k \log n)$ for edit errors.  However, the running time of the protocol is exponential in $k$. The main question left there is whether one can design a document exchange protocol that is both communication efficient and time efficient.

There has been considerable progress afterwards \cite{CormodePSV00}, \cite{irmak2005improved}, \cite{Jowhari2012EfficientCP}. Specifically, Irmak et al. \cite{irmak2005improved} gave a randomized protocol that achieves sketch size $O(k \log(\frac{n}{k}) \log n)$ and running time $\tilde{O}(n)$. Independently, Jowhari \cite{Jowhari2012EfficientCP} also obtained a randomized protocol with sketch size $O(k \log^2 n \log^* n)$ and running time $\tilde{O}(n)$. A recent work by Chakraborty et al. \cite{Chakraborty2015LowDE} introduced a clever randomized embedding from the edit distance metric to the Hamming distance metric, and thus obtained a protocol with sketch size $O(k^2 \log n)$ and running time $\tilde{O}(n)$. Using the embedding in \cite{Chakraborty2015LowDE}, Belazzougui and Zhang \cite{BelazzouguiZ16} gave an improved randomized protocol with sketch size $O(k (\log^2 k+\log n))$ and running time $\tilde{O}(n+\poly(k))$, where the sketch size is asymptotically optimal for $k =2^{O(\sqrt{\log n})}$.

All of the above protocols, except the exponential time protocol of Orlitsky \cite{185373}, are however randomized. In practice, a deterministic protocol is certainly more useful than a randomized one. Thus one natural and important question is whether one can construct a deterministic protocol for document exchange with small sketch size (e.g., polynomial in $k \log n$) and efficient computation. This question is also important for applications in error correcting codes, since a randomized document exchange protocol is not very useful in designing such codes. It turns out that this question is quite tricky, and no such deterministic protocols are known even for $k>1$ until the work of Belazzougui \cite{Belazzougui2015EfficientDS} in 2015, where he gave a deterministic protocol with sketch size $O(k^2 + k \log^2 n)$ and running time  $\tilde{O}(n)$.

\paragraph{Error correcting codes.} Error correcting codes are fundamental objects in both theory and practice. Starting from the pioneering work of Shannon, Hamming and many others, error correcting codes have been intensively studied in the literature. This is true for both standard Hamming errors such as symbol corruptions and erasures, and edit errors such as insertions and deletions. While the study of codes against standard Hamming errors has been a great success, leading to a near complete knowledge and a powerful toolbox of techniques together with explicit constructions that match various bounds, our understanding of codes for insertion and deletion errors (insdel codes for short) is still rather poor. Indeed, insertion and deletion errors are strictly more general than Hamming errors, and the study of codes against such errors has resisted progress for quite some time, as demonstrated by previous work which we discuss below.

Since insertion and deletion errors are strictly more general than Hamming errors, all the upper bounds on the rate of standard codes also apply to insdel codes.\ Moreover, by using a similar sphere packing argument, similar lower bounds on the rate (such as the Gilbert-Varshamov bound) can also be shown. In particular, one can show (e.g., \cite{Levenshtein66}) that for binary codes, to encode a message of length $n$ against $k$ insertion and deletion errors with $k \leq n/2$, the optimal number of redundant bits is $\Theta(k \log (\frac{n}{k}))$; and to protect against $\eps$ fraction of insertion and deletion errors, the optimal rate of the code is $1-\Theta(\eps \log (\frac{1}{\eps}))$. On the other hand, if the alphabet of the code is large enough, then one can potentially recover from an error fraction approaching $1$ or achieve the singleton bound: a rate $1-\eps$ code that can correct $\eps$ fraction of insertion and deletion errors.

However, achieving these goals have been quite challenging.\ In 1966, Levenshtein \cite{Levenshtein66} first showed that the Varshamov-Tenengolts code \cite{VT65} can correct one deletion with roughly $\log n$ redundant bits, which is optimal. Since then many constructions of insdel codes have been given, but all constructions are far from achieving the optimal bounds. In fact, even correcting two deletions requires $\Omega(n)$ redundant bits, and even the first explicit asymptotically good insdel code (a code that has constant rate and can also correct a constant fraction of insertion and deletion errors) over a constant alphabet did not appear until the work of  Schulman and Zuckerman in 1999 \cite{796406}, who gave such a code over the binary alphabet. We refer the reader to the survey by Mercier et al. \cite{MercierBT10} for more details about the extensive research on this topic.

In the past few years, there has been a series of work trying to improve the situation for both the binary alphabet and larger alphabets. Specifically, for larger alphabets, a line of work by Guruswami et.\ al \cite{7835185}, \cite{7541373}, \cite{BukhV16} constructed explicit insdel codes that can correct $1-\eps$ fraction of errors with rate $\Omega(\eps^5)$ and alphabet size $\poly(1/\eps)$; and for a fixed alphabet size $t \geq 2$ explicit insdel codes that can correct $1-\frac{2}{t+1}-\eps$ fraction of errors with rate $(\eps/t)^{\poly(1/\eps)}$. These works aim to tolerate an error fraction approaching $1$ by using a sufficiently large alphabet size. Another line of work by Haeupler et al \cite{haeupler2017synchronization}, \cite{HS17c}, \cite{CHLSW18} introduced and constructed a combinatorial object called \emph{synchronization string}, which can be used to transform standard error correcting codes into insdel codes, at the price of increasing the alphabet size.\ Using explicit constructions of synchronization strings, \cite{haeupler2017synchronization} achieved explicit insdel codes that can correct $\delta$ fraction of errors with rate $1-\delta-\eps$ (hence approaching the singleton bound), although the alphabet size is exponential in $\frac{1}{\eps}$.

For the binary alphabet, which is the focus of this paper, it is well known that no code can tolerate an error fraction approaching $1$ or achieve the singleton bound. Instead, the major goal here is to construct explicit insdel codes for some small  fraction ($\eps$) or some small number ($k$) of errors that can achieve the optimal rate of $1-\Theta(\eps \log (\frac{1}{\eps}))$ or the optimal redundancy of $\Theta(k \log (\frac{n}{k}))$, which is analogous  to achieving the Gilbert-Varshamov bound for standard error correcting codes. Slightly less ambitiously, one can ask to achieve redundancy $O(k \log n)$, which is optimal when $k \leq n^{1-\alpha}$ for any constant $\alpha>0$, and easy to achieve in the case of Hamming errors by using Reed-Solomon codes. In this context, Guruswami et.\ al \cite{7835185}, \cite{7541373} constructed explicit insdel codes that can correct $\eps$ fraction of errors with rate $1-\tilde{O}(\sqrt{\eps})$, which is the best possible by using code concatenation.\ For any fixed constant $k$, another work by Brakensiek et.\ al \cite{8022906} constructed an explicit insdel code that can encode an $n$-bit message against $k$ insertions/deletions with $O(k^2 \log k \log n)$ redundant bits, which is asymptotically optimal when $k$ is a fixed constant. We remark that the construction in \cite{8022906} only works for constant $k$, and does not give anything when $k$ becomes larger (e.g., $k=\log n$). Finally, using his deterministic document exchange protocol, Belazzougui \cite{Belazzougui2015EfficientDS} constructed an explicit insdel code that can encode an $n$-bit message against $k$ insertions/deletions with $O(k^2 + k \log^2 n)$ redundant bits. In summary, there remains a huge gap between the known constructions and the optimal bounds in the case of a binary alphabet. In particular, even achieving $O(k \log n)$ redundancy has been far out of reach.

\subsection{Our results}

In this paper we significantly improve the situation for both document exchange and binary insdel codes. Our new constructions of explicit insdel codes are actually almost optimal for a wide range of error parameters $k$. First, we have the following theorem which gives an improved deterministic document exchange protocol.

\begin{theorem}\label{thm1}
There exists a single round deterministic protocol for document exchange with communication complexity (sketch length) $O(k \log^2 \frac{n}{k}) $, time complexity $\poly(n)$, where $n$ is the length of the string and $k$
is the edit distance upper bound.
\end{theorem}

Note that this theorem significantly improves the sketch size of the deterministic protocol in \cite{Belazzougui2015EfficientDS}, which is $O(k^2 + k \log^2 n)$. In particular, our protocol is interesting for $k$ up to $\Omega(n)$ while the protocol in \cite{Belazzougui2015EfficientDS} is interesting only for $k < \sqrt{n}$.

We can use this theorem to get improved binary insdel codes that can correct $\eps$ fraction of errors.

\begin{theorem} \label{thm2}
There exists a constant $0<\alpha<1$ such that for any  $0<\eps \leq \alpha$ there exists an explicit family of binary error correcting codes with codeword length $n$ and message length $m$, that can correct up to $k=\eps n$ edit errors with rate $m/n=1- O(\eps \log^2 \frac{1}{\eps})$. 
\end{theorem}

Note that the rate of the code is $1- O(\eps \log^2 \frac{1}{\eps})=1-\widetilde{O}(\eps)$, which is optimal up to an extra $\log(\frac{1}{\eps})$ factor. This significantly improves the rate of $1-\widetilde{O}(\sqrt{\eps})$ in \cite{7835185}, \cite{7541373}.

For the general case of $k$ errors, the document exchange protocol also gives an insdel code.

\begin{theorem} \label{thm3}
For any $n, k \in \mathsf{N}$ with $k \leq n/4$, there exists an explicit binary error correcting code with message length $n$, codeword length $n+O(k \log^2 \frac{n}{k})$ that can correct up to $k$ edit errors. 

\end{theorem}

When $k$ is small, e.g., $k=n^{\alpha}$ for some constant $\alpha<1$, the above theorem gives  $O(k \log^2 n)$ redundant bits. Our next theorem shows that we can do better, and in fact we can achieve redundancy $O(k \log n)$, which is asymptotically optimal for small $k$. 

\begin{theorem}\label{thm4}
For any $n, k \in \mathsf{N}$, there exists an explicit binary error correcting code with message length $n$, codeword length $n+O(k \log n)$ that can correct up to $k$ edit errors. 

\end{theorem}

Note that in this theorem, the number of redundant bits needed is $O(k\log n)$, which is asymptotically optimal for $k \leq n^{1-\alpha}$, any constant $0<\alpha<1$. This significantly improves the construction in \cite{8022906}, which only works for constant $k$ and has redundancy $O(k^2 \log k \log n)$, and the construction in \cite{Belazzougui2015EfficientDS}, which has redundancy $O(k^2 + k \log^2 n)$. In fact, this is the first explicit construction of binary insdel codes that have optimal redundancy for a wide range of error parameters $k$, and this brings our understanding of binary insdel codes much closer to that of standard binary error correcting codes.

\begin{remark}
In all our insdel codes, both the encoding function and the decoding function run in time $\poly(n)$. 
\end{remark}

\paragraph{Independent work.} In a recent independent work \cite{haeupler2018optimal}, Haeupler also studied the document exchange problem and binary error correcting codes for edit errors. Specifically, he also obtained a deterministic document exchange protocol with sketch size $O(k \log^2 \frac{n}{k}) $, which leads to an error correcting code with the same redundancy, thus matching our theorems~\ref{thm1}, ~\ref{thm2} and ~\ref{thm3}. He further gave a randomized document exchange protocol that has optimal sketch size $O(k \log \frac{n}{k}) $. However, for small $k$, our Theorem~\ref{thm4} gives a much better error correcting code. Our code is optimal for $k \leq n^{1-\alpha}$, any constant $0<\alpha<1$, and thus better than the code given in \cite{haeupler2018optimal}. 

\subsection{Overview of the techniques}
In this section we provide a high level overview of the ideas and techniques used in our constructions. We start with the document exchange protocol.

\paragraph{Document exchange.} Our starting point is the randomized protocol by Irmak et al. \cite{irmak2005improved}, which we refer to as the IMS protocol. The protocol is one round, but Alice's algorithm to generate the message proceeds in $L=O(\log(\frac{n}{k}))$ levels. In each level Alice computes some sketch about her string $x$, and her final message to Bob is the concatenation of the sketches. After receiving the message, Bob's algorithm also proceeds in $h$ levels, where in each level he uses the corresponding sketch to recover part of $x$. 

More specifically, in the first level Alice divides her string into $\Theta(k)$ blocks where each block has size $O(\frac{n}{k})$, and in each subsequent level every block is divided evenly into two blocks, until the final block size becomes $O(\log n)$. This takes $O(\log(\frac{n}{k}))$ levels. Using shared randomness, in each level Alice picks a set of random hash functions, one for each block which outputs $O(\log n)$ bits, and computes the hash values. In the first level, Alice's sketch is just the concatenation of the $O(k)$ hash values. In all subsequent levels, Alice obtains the sketch in this level by computing the redundancy of a systematic error correcting code (e.g., the Reed-Solomon code) that can correct $O(k)$ erasures and symbol corruptions, where each symbol has $O(\log n)$ bits (the hash value). Note that this sketch has size $O(k \log n)$ and thus the total sketch size is $O(k \log n \log(\frac{n}{k}))$.

On Bob's side, he always maintains a string $\tilde{x}$ which is the partially corrected version of $x$. Initially $\tilde{x}$ is the empty string, and in each level Bob tries to use his string $y$ to fill $\tilde{x}$. This is done as follows. In each level Bob first tries to recover all the hash values of Alice in this level (notice that the hash values of the first level are directly sent to Bob). Suppose Bob has successfully recovered all the hash values, Bob then tries to match every block of Alice's string in this level with one substring of the same length in his string $y$, by finding such a substring with the same hash value. We say such a match is bad if the substring Bob finds is not the same as Alice's block (i.e., a hash collision). The key idea here is that if the hash functions output $O(\log n)$ bits, and they are chosen independently randomly, then with high probability a bad match only happens if the substring Bob finds contains at least one edit error. Moreover, Bob can find at least $l_i-k$ matches, where $l_i$ is the number of blocks in the current level. Bob then uses the matched substrings to fill the corresponding blocks in $\tilde{x}$, and leaves the unmatched blocks blank. From the above discussion, one can see that there are at most $k$ unmatched blocks and at most $k$ mismatched blocks. Therefore in the next level when both parties divide every block evenly into two blocks, $x$ and $\tilde{x}$ have at most $4k$ different blocks. This implies that there are also at most $4k$ different hash values in the next level, and hence Bob can correctly recover all the hash values of Alice using the redundancy of the error correcting code.  

Our deterministic protocol for document exchange is a derandomized version of the IMS protocol, with several modifications. First, we observe that the IMS protocol as we presented above, can already be derandomzied. This is because that to ensure a bad match only happens if the substring Bob finds contains at least one edit error, we in fact just need to ensure that under any hash function, no \emph{block} of $x$ can have a collision with a \emph{substring} of the same length in $x$ itself. We emphasize one subtle point here: Alice's hash function is applied to a block of her string $x$, while when trying to fill $\tilde{x}$, Bob actually checks every substring of the string $y$. Therefore we need to consider hash collisions between blocks of $x$ and substrings of $x$. If the hash functions are chosen independently uniformly, then such a collision happens with probability $1/\poly(n)$, and thus by a union bound with high probability no collision happens. However, notice that if we write out the outputs of all hash functions on all inputs, then any collision is only concerned with two outputs which consists of $O(\log n)$ bits. Thus it's enough to use $O(\log n)$-wise independence to generate these outputs. To further save the random bits used, we can instead use an almost $\kappa$-wise independent sample space with $\kappa=O(\log n)$ and error $\eps=1/\poly(n)$. Using for example the construction by Alon et. al. \cite{alon1992simple}, this results in a total of $O(\log n)$ random bits (the seed), and thus Alice can exhaustively search for a fixed set of hash functions in polynomial time. Now in each level, Alice's sketch will also include the specific seed that is used to generate the hash functions, which has $O(\log n)$ bits. Note this only adds $O(\log n \log \frac{n}{k})$ to the final sketch size. Bob's algorithm is essentially the same, except now in each level he needs to use the seed to compute the hash functions. 

The above construction gives a deterministic document exchange protocol with sketch size $O(k \log n \log \frac{n}{k})$, but our goal is to further improve this to $O(k \log^2 \frac{n}{k})$. The key idea here is to use a relaxed version of hash functions with nice ``self matching" properties. To motivate our construction, first observe that in each level, when Bob tries to match every block of Alice's string with one substring of the same length in his string $y$, it is not only true that Bob can find a matching of size at least $l_i-k$ (where $l_i$ is the number of blocks in this level), but also true that Bob can find a \emph{monotone} matching of at least this size. A monotone matching here means a matching that does not have edges crossing each other. In this monotone matching, there are at most $k$ bad matches caused by edit errors, and thus there exists a \emph{self matching} between $x$ and itself with size at least $l_i-2k$. In the previous construction, we in fact ensure that all these $l_i-2k$ matches are correct. To achieve better parameters, we instead relax this condition and only require that at most $k$ of these self matches are bad. Note if this is true then the total number of different blocks between $x$ and $\tilde{x}$ is still $O(k)$ and we can again use an error correcting code to send the redundancy of hash values in the next level. 

This relaxation motivates us to introduce \emph{$\eps$-self matching hash functions}, which is similar in spirit to $\eps$-self matching strings introduced in \cite{haeupler2017synchronization}. Formally, we have the following definitions.

\begin{definition}(monotone matching)
For every $n, n', t,  p, q  \in \mathbb{N}, q\leq p$, any hash functions $h_1, h_2, \ldots, h_{n'}$ where $\forall i\in [n'], h_i: \{0,1\}^{p } \rightarrow \{0,1\}^{q  }$, given two strings $x\in (\{0,1\}^{p})^{n'}$ and $y \in \{0,1\}^{n}$, a monotone matching of size $t$ between $x, y$ under hash functions $h_1, \ldots, h_{n'}$  is a sequence of pairs of indices $w = ((i_1, j_1), (i_2, j_2), \ldots, (i_t, j_t)) \in ([n'] \times [n])^t$ s.t. 
$i_1 < i_2 < \cdots < i_t$,   $j_{1}+p-1 < j_{2}, \ldots, j_{t-1}+p-1 < j_{t}, j_{t}+p-1 \leq n$ and $ \forall l\in [t], h_{i_l}(x[i_l]) = h_{i_l}(y[j_l,  j_l+p-1] )$. When $h_1, \ldots, h_{n'}$ are clear from the context, we simply say that $w$ is a monotone matching between $x, y$. 

For $l\in [t]$, if $x[i_l] = y[j_l,  j_l+p-1] $, we say $(i_l, j_l)$ is a good match, otherwise we say it is a bad match. We say $w$ is a correct matching if all matches in $w$ are good. We say $w$ is a completely wrong matching is all matches in $w$ are bad.

If $x$ and $y$ are the same in terms of their binary expression, then $w$ is called a self-matching.
\end{definition}

For simplicity, in the rest of the paper, when we say a matching $w$ we always mean a monotone matching.

\begin{definition}($\eps$-self matching hash function)
Let $p, q, n, n' \in \mathbb{N}$ be such that $n=n' p$. For any $0 < \eps < 1$ and $x \in (\{0,1\}^{p})^{n'}$, we say that a sequence of hash functions $h_1, h_2, \ldots, h_{n'}$ where $\forall i\in [n'], h_i: \{0,1\}^{p  } \rightarrow \{0,1\}^{q  }$ is a sequence of $\eps$-self matching hash functions with respect to $x$, if any matching between $x$ and $y \in \{0,1\}^{n}$ under $h_1, h_2, \ldots, h_n$, where $y$ is the binary expression of $x$, has at most $\eps n$ bad matches.
\end{definition}

The advantage of using $\eps$-self matching hash functions is that the output range of the hash functions can be reduced. Specifically, we can show that a sequence of $\eps$-self matching hash functions exists with output range $\poly(1/\eps)$ (i.e., $O(\log(1/\eps))$ bits) when the block size is at least $c \log(1/\eps)$ bits for some constant $c>1$. Furthermore, we can generate such a sequence of $\eps$-self matching hash functions with high probability by again using an almost $\kappa$-wise independent sample space with $\kappa=O(k b_i)$, where $b_i$ is the current block length, and error $\eps=1/\poly(n)$. The idea is that in a monotone matching with $\eps n$ bad matches, we can divide the matching gradually into small intervals such that at least one small interval will have the same fraction of bad matches. Thus in order to ensure the $\eps$-self matching property we just need to make sure every small interval does not have more than $\eps$ fraction of bad matches, and this is enough by using the almost $\kappa$-wise independent sample space. 

As discussed above, we need to ensure that there are at most $k$ bad matches in a self matching, thus we set $\eps=\frac{k}{n}$. Consequently now the output of the hash functions only has $O(\log(n/k))$ bits instead of $O(\log n)$ bits. Now in each level, in order to get optimal sketch size, instead of using the Reed-Solomon code we will be using an Algebraic Geometric code \cite{hoholdt1998algebraic} which has redundancy $O(k \log \frac{n}{k})$. The almost $\kappa$-wise independent sample space in this case again uses only $O(\log n)$ random bits, so in each level Alice can exhaustively search the correct hash functions in polynomial time and include the $O(\log n)$ bits of description in the sketch. This gives Alice's algorithm with total sketch size $O(k \log^2 \frac{n}{k})$. On Bob's side, we need another modification: in each level after Bob recovers all the hash values, instead of simply searching for a match for every block, Bob runs a dynamic programming to find the longest monotone matching between his string $y$ and the sequence of hash values. He then fills the blocks of $\tilde{x}$ using the corresponding substrings of matched blocks. 

\paragraph{Error correcting codes.} Our deterministic document exchange protocol can be used to directly give an insdel code for $k$ edit errors. The idea is that to encode an $n$-bit message $x$, we can first compute a sketch of $x$ with size $r$, and then encode the small sketch using an insdel code against $4k$ edit errors. Since the sketch size is larger than $k$, we can use an asymptotically good code such as the one by Schulman and Zuckerman \cite{796406}, which results in an encoding size of $n_0=O(r)$. The actual encoding of the message is then the original message concatenated with the encoding of the sketch. 

To decode, we can first obtain the sketch by looking at the last $n_0-k$ bits of the received string. The edit distance between these bits and the encoding of the sketch is at most $4k$, and thus we can get the correct sketch from these bits. Now we look at the bits of the received string from the beginning to index $n+k$. The edit distance between these bits and $x$ is at most $3k$, thus if $r$ is a sketch for $3k$ edit errors then we will be able to recover $x$ by using $r$. This gives our first insdel code with redundancy $O(k \log^2 \frac{n}{k})$.

We now describe our second insdel code, which has redundancy $O(k \log n)$ and uses many more interesting ideas. The basic idea here is again to compute a sketch of the message and encode the sketch, as we described above. However, we are not able to improve the sketch size of $O(k \log^2 \frac{n}{k})$ in general. Instead, our first observation here is that if the $n$-bit message is a \emph{uniform random} string, then we can actually do better. Thus, we will first describe how to come up with a sketch of size $O(k \log n)$ for a uniform random string, and then use this to get an encoding for any given $n$-bit string.

To warm up, we first explain a simple algorithm to compute a sketch of size $O(k \log^2 n)$ in this case. A uniform random string has many nice properties. In particular, for some $B=O(\log n)$, one can show that with high probability, every length $B$ substring in a uniform random string is \emph{distinct}. If this property holds (which we refer to as the \emph{$B$-distinct property}), then Alice can compute a sketch as follows. First create a vector of length $2^B=\poly(n)$, where in each entry indexed by the string $s \in \bit^B$, Alice looks at the substring $s$ in $x$ and record the bit left to it and the bit right to it. We need three special symbols, one to indicate the case of no left bit, one to indicate the case of no right bit, and one to indicate the case that there is no such string $s$ in $x$. Thus it is enough to use an alphabet of size $8$ for each entry. Similarly Bob can create a vector $V'$ from his string $y$. Notice that the entries in $V$ have no collisions since we assume that Alice's string is $B$-distinct, while some entries in $V'$ may have collisions due to edit errors, in which case Bob just treats it as there is no corresponding substring in $y$. One can then show that $V$ and $V'$ differ in at most $O(k B)=O(k \log n)$ entries. Now Alice can use the Reed-Solomon code to send a redundancy of size $O(k \log^2 n)$, and Bob can recover the correct $V$. Bob can then recover the string $x$ by picking an entry in $V$ and growing the string on both ends gradually until obtaining the full string $x$.

We now show how to reduce the sketch size. In the above approach,  $V$ and $V'$ can differ in $O(k \log n)$ entries since Alice is looking at \emph{every} substring of length $B$ in $x$. To improve this, instead we will have Alice first partition her string into several blocks, and then just look at the substrings corresponding to each block. Ideally, we want to make sure that each block has length at least $B$ so that again all blocks are distinct. Alice then creates the vector $V$ of length $2^B$ by using the $B$-prefix of each block as the index in $V$, and for each entry in $V$ Alice will record some information. Bob will do the same thing using his string $y$ to create another vector $V'$, and we will argue that $V$ and $V'$ do not differ much so Alice can send some redundancy information to Bob, and Bob can recover the correct $V$ based on $V'$ and the redundancy information. 

However, the partitions need to be done carefully. For example, we cannot just partition both strings sequentially into blocks of size some $T \geq B$, since if so then a single insertion/deletion at the beginning of $x$ could result in the case where all blocks of $x$ and $y$ are distinct. Instead, we will choose a specific string $p$ with length $s$, for some parameter $s$ that we will choose appropriately. We call this string $p$ a \emph{pattern}, and we will use this pattern to divide the blocks in $x$ and $y$. More specifically, in our construction we will simply choose $p=1 \circ 0^{s-1}$, the string with a $1$ followed by $s-1$ $0$'s. We use $p$ to divide the string $x$ as follows. Whenever $p$ appears as a substring in $x$, we call the index corresponding to the bit of $p$ in $x$ a \emph{p-split point}. The set of split points then naturally gives a partition of the string $x$ (and also $y$ into blocks). We note that the idea of using patterns and split points is also used in \cite{8022906}. However, there the construction uses $2k+1$ patterns and takes a majority vote, which is why the sketch has a $k^2$ factor. Here instead we use a single pattern, and thus we can avoid the $k^2$ factor.

We now describe how to choose the pattern length $s$. For the blocks of $x$ and $y$ obtained by the split points, we certainly do not want the block size to be too large. This is because if there are large blocks then an adversary can create errors in these blocks, and large blocks need longer sketches to recover from error. At the same time, we do not want the block size to be too small either. This is because if the block size is too small, then some of the blocks may actually be the same, while we would like to keep all blocks of $x$ to be distinct. In particular, we would like to keep the size of every block in $x$ to be at least $B$. If $x$ is a uniform random string, then the pattern $p$ appears with probability $2^{-s}$ and thus the expected distance between two consecutive appearances of $p$ is $2^s$. Moreover, one can show that with high probability any interval of length some $O(s 2^s \log n)$ contains a $p$-split point. We will call this \emph{property} 1. If this property holds then we can argue that every block of $x$ has length at most $O(s 2^s \log n)$.  

To ensure that each block is not too short, we simply look at a $p$-split point and the immediate next $p$-split point after it. If the distance between these two split points is less than $2^s/2$ then we just ignore the first $p$-split point. In other words, we change the process of dividing $x$ into blocks so that we will only use a $p$-split point if the next $p$-split point is at least $2^s/2$ away from it, and we call such a $p$-split point a good $p$-split point. We again show that if $x$ is uniform random, then with high probability every block of length some $O(2^s \log n)$ contains a good $p$-split point. We call this \emph{property} 2. Combined with the previous paragraph, we can now argue that with high probability every interval of length some $O(s 2^s \log n)+O(2^s \log n)=O(s 2^s \log n)$ will contain a $p$-split point that we will choose. By setting $s=\log \log n+O(1)$, we can ensure that every block of $x$ has length at least $B=O(\log n)$ and at most $O(s 2^s \log n)=\poly \log(n)$.

Now for each block of $x$, Alice creates an entry in $V$ indexed by the $B$-prefix of this block. The entry contains the length of this block and the $B$-prefix of the next block, which has total size $O(\log n)$. Similarly Bob also creates a vector $V'$. We show that our approach of choosing $p$-split points can ensure that $V$ and $V'$ differ in at most $O(k)$ entries, thus Alice can send a string with $O(k \log n)$ redundant bits to Bob (using the Reed-Solomon code) and Bob can recover the correct $V$. The structure of $V$ guarantees that Bob can learn the correct order of the $B$-prefix of all Alice's blocks, and their lengths. At this point Bob will again try to fill a string $\tilde{x}$, where for each of Alice's block Bob searches for a substring with the same $B$-prefix and the same length. We show that after this step, at most $O(k)$ blocks in $\tilde{x}$ are incorrectly filled or missing. This completes state 1 of the sketch.

We now move to stage 2, where Bob correctly recovers the at most $O(k)$ blocks in $\tilde{x}$ that are incorrectly filled or missing. One way to do this is by noticing that every block has size at most $\poly\log(n)$, thus we can use a deterministic IMS protocol as we described before, which will last $O(\log \log n)$ levels and thus have sketch size $O(k \log \frac{n}{k} \log \log n$) (since we start with block size $\poly\log(n)$). However, our goal is to do better and achieve sketch size $O(k \log n)$.

To achieve this, we modify the IMS protocol so that in stage 2, in each level we are not dividing every block into $2$ smaller blocks. Instead, we divide every block evenly into $O(\log^{0.4} n)$ smaller blocks. We continue this process until the final block size is $O(\log n)$, and thus this only takes $O(1)$ levels. In each level, we will do something similar to our deterministic document exchange protocol: Alice sends a description of a sequence of hash functions to Bob, together with some redundancy of the hash values. Bob recovers the correct hash values and uses a dynamic programming to find the longest monotone matching between substrings of $y$ and the hash values. Bob then tries to fill the blocks of $\tilde{x}$ by using the matched substrings. 

In order for the above approach to work, we need to ensure three things. First, Alice's description of the hash functions should be short, ideally only $O(\log n)$ bits. Second, the redundancy of the hash values only uses $O(k \log n)$ bits. Finally, after Bob recovers the hash values, the matching he finds contains at most $O(k)$ unmatched blocks and mismatched blocks. For the second issue, we design the hash functions so that the outputs only have $O(\log^{0.5} n)$ bits. One issue here is that if in some level there are at most $O(k)$ unmatched blocks and mismatched blocks, then in the next level after dividing there may be $O(k \log^{0.4} n)$ unmatched blocks and mismatched blocks. However we observe that these blocks are actually concentrated (i.e., they are smaller blocks in a larger block of the previous level). Thus we can pack all the hash values of the $O(\log^{0.4} n)$ smaller blocks together into a package, and the total size of the hash values is $O(\log^{0.4} n) \cdot O(\log^{0.5} n) = O(\log n)$. Now we can show that again there are at most $O(k)$ different packages between Alice's version and Bob's version, thus it is still enough to use $O(k \log n)$ bits of redundancy. 

The third issue and the first issue are actually related, and require new ideas. Specifically, we cannot use the $\eps$-self matching hash functions as we discussed earlier, since that would require the outputs of the hash functions to have $O(\log(1/\eps))=O(\log (n/k))$ bits, while we can only afford $O(\log^{0.5} n)$ bits. To solve this problem, we strengthen the notion of $\eps$-self matching hash functions and introduce \emph{$\eps$-synchronization hash functions}, similar in spirit to the notion of $\eps$-synchronization strings introduced in \cite{haeupler2017synchronization}. Specifically, we have the following definition.

\begin{definition}
Let $T, n' , B,  R \in \mathbf{N}$ be such that $T \geq B$, and $0 < \eps < 1$.  Let $x$ be a string of length $n = n' T$, and $x_T = (x[1, T], x[T+1, 2T], \dots, x[(n'-1)T+1, n])$ be a partition of $x$ into blocks of size $T$. 
	Let $\Phi = (\Phi[1], \Phi[2], \dots, \Phi[n'])$ be a sequence of functions,
	where for any $k \in [n']$, $\Phi[k]$ is a function from $\bit^B$ to $\{0, 1\}^R$.
	
	For a string $y$ of length $m$, and some indices $0 \le l_1 < r_1 \le n'$, $0 \le l_2 < r_2 \le m$,
	let $\mathsf{MATCH}_{\Phi}(x_T(l_1, r_1], y(l_2, r_2])$ denote the size of the maximum matching between $x_T(l_1, r_1]$
	and $y(l_2, \min(r_2+T, m+1))$ under $\Phi(l_1, r_1]$. We say $\Phi$ is a sequence of  $\eps$-synchronization hash functions with respect to $x$, if it satisfies the following properties:
	\begin{itemize}
		\item For any three integers $i, j, k$ where $i < Tj$ and $j < k$, denote $l_1 = k - j$ and $l_2 = Tj - i$.
		\begin{equation} 
		\mathsf{MATCH}_{\Phi}(x_T(j, k], x(i, Tj]) < \eps \left(
		l_1 + \frac{l_2}{T}
		\right) 
		\end{equation}
		
		\item For any three integers $i, j, k$ where $i < j$ and $k > T(j-1)+1$, denote $l_1 = j - i$ and $l_2 = k - T(j-1) - 1$.
		\begin{equation}
		\mathsf{MATCH}_{\Phi}(x_T(i, j], x(T(j-1)+1, k]) < \eps \left(l_1 + \frac{l_2}{T}
		\right)
		\end{equation}
	\end{itemize}
\end{definition}

It can be seen that $\eps$-synchronization hash functions are indeed stronger than $\eps$-self matching hash functions, and we show that even a sequence of $\eps$-synchronization hash functions with $\eps=\Omega(1)$ is enough to guarantee that there are at most $O(k)$ unmatched blocks and mismatched blocks (in fact, the number of unmatched blocks and mismatched blocks is bounded by $\frac{1+2\eps}{1-2\eps}k$). A similar property for $\eps$-synchronization strings is also shown in \cite{haeupler2017synchronization}.

Our construction of the $\eps$-synchronization hash functions actually consists of two parts. In the first part, Alice generates a sequence of hash functions that ensures the $\eps$-synchronization property holds for relatively large intervals (i.e., when $l_1 + \frac{l_2}{T}$ is large). We show that this sequence of hash functions can again be generated by using an almost $\kappa$-wise independent sample space which uses $O(\log n)$ random bits (the argument is similar to that of $\eps$-self matching hash functions). Thus Alice can exhaustively search for a fixed set of hash functions in polynomial time, and send Bob the description using $O(\log n)$ bits. The output of these hash functions has $O(\log^{0.5} n)$ bits. To protect the small intervals, we compute a hash function for each block such that when the inputs are restricted to length $B$ substrings in a small interval, this function is \emph{injective}. Since all substrings of length $B$ are distinct in $x$, this ensures that the outputs of the same hash function within a small interval are also distinct, and thus the $\eps$-synchronization property also holds for small intervals. The output of these hash functions has $O(\log \log n)$ bits. We show that these functions can be constructed deterministically using an almost $\kappa$-wise independent sample space which uses $O(\log n \log n)$ random bits, and hence also has description size $O(\log n \log n)$. Moreover the set of functions computed by Alice and the set of functions computed by Bob have at most $O(k)$ differences (after packing each successive $\log^{0.6} n$ such functions together, which only needs $O(\log^{0.6} n\log n \log n)$ bits), and thus Alice can again send $O(k \log n)$ redundant bits to Bob and Bob can recover the correct set of hash functions. The final sequence of $\eps$-synchronization hash functions is then the combination of these two sets of functions, where the output has $O(\log^{0.5} n+O(\log \log n))=O(\log^{0.5} n)$ bits. This concludes our algorithm to generate a sketch of size $O(k \log n)$ for a uniform random string.

We now turn to solve the issue of assuming a uniform random string $x$. Put simply, our idea is that given any arbitrary string $x$, we first compute the XOR of $x$ and a \emph{pseudorandom} string $z$ (a mask), to turn $x$ into a pseudorandom string as well. We will use $O(\log n)$ random bits to generate $z$ and show that with high probability over the random bits used, the XOR of $x$ and $z$ satisfies the $B$-distinct property, as well as property 1 and property 2. For this purpose, we construct three pseudorandom generators (PRGs), one for each property. The $B$-distinct property can be ensured by again using an almost $\kappa$-wise independent sample space which uses $O(\log n)$ random bits. For property 1, we divide the string of length $n$ into blocks of length $T = O(s2^s \log n)$ and generate the same mask for every block. Within each block, we can view it as a sequence of $t=O(\log n)$ sub-blocks each of length  $s2^s$. For each sub-block, the test of whether it contains the pattern $p$ can be realized as a DNF with size $\poly \log n$, so we can use a PRG for DNF to fool this test with constant error, and this has length $\poly \log \log (n)$. We then use a random walk on a constant-degree expander graph of length $t$ to generate the full mask, which guarantees that the pattern occurs with probability $1-1/\poly(n)$. A union bound now shows that property 1 holds with high probability, and the PRG has seed length $O(\log n)$. For property 2, we can essentially do the same thing, i.e., divide the string of length $n$ into blocks of length $T = O(2^s \log n)$ and generate the same mask for every block. For each block, again we use a PRG for DNF together with a random walk on a constant-degree expander graph to get a PRG with seed length $O(\log n)$. Finally we take the XOR of the outputs of the three PRGs, and we show that with high probability all three properties hold for $x \oplus z$. 

Since the PRG has seed length $O(\log n)$, again we can exhaustively search for a fixed $z$ that works for the string $x$, and $z$ can be described by $O(\log n)$ bits. The encoding of $x$ is then $x \oplus z$, together with an encoding of the concatenation of the description of $z$ and the sketch. For decoding, one can first recover $x \oplus z$ and then recover $x$ by using the description of $z$ and the PRG. 

\subsection{More on $\eps$-synchronization hash functions}

Our definition and construction of $\eps$-synchronization hash functions turn out to have several tricky issues. First of all, there may be other possible definitions of such hash functions, but for our application the current definition is the most suitable one. Second, in our construction of the $\eps$-synchronization hash functions, we crucially use the fact that the string $x$ has the $B$-distinct property. This is because if $x$ does not have this property, then when we try to find a maximum monotone matching, it may be the case that every pair in the matching is illy matched (e.g., the strings corresponding to the pairs are in different positions but the strings themselves are the same) and this will cause a problem in our analysis. This problem may be solved by modifying the definition of $\eps$-synchronization hash functions, but it will potentially result in a larger output range of the hash functions (e.g., $\Omega(\log n)$ bits), which we cannot afford. 

Finally, our construction of the $\eps$-synchronization hash functions consists of two separate sets of hash functions, one for large intervals and one for small intervals. In the construction of hash functions for small intervals, we again crucially use the fact that he string $x$ has the $B$-distinct property, so that for each block both parties can compute a hash function and Alice can just send the redundancy. We note that here one cannot simply apply the (deterministic) Lov\'asz Local Lemma as in  \cite{haeupler2017synchronization}, \cite{HS17c}, \cite{CHLSW18}, since this will result in a very long description of the hash functions, and Alice cannot just send it to Bob.

\paragraph{Organization of this paper}
In Section \ref{sec:prelim} we introduce some notation and basic techniques used in this paper. In Section \ref{sec:determdocexc} we give the deterministic protocol for document exchange.\ In Section \ref{sec:randdocexc} we give a protocol for document exchange of a uniform random string.\ In Section \ref{sec:InsdelCode} we construct error correcting codes for edit errors using results from previous sections. Finally we conclude with some open problems in Section \ref{sec:open}.

\section{Preliminaries} \label{sec:prelim}
\subsection{Notation}
Let $\Sigma$ be an alphabet (which can also be a set of strings).
For a string $x\in \Sigma^*$,
\begin{enumerate}
\item $|x|$ denotes the length of the string.
\item $x[i,j]$ denotes the substring of $x$ from position $i$ to position $j$ (Both ends included).
\item $x[i]$ denotes the $i$-th symbol of $x$.
\item $x\circ x'$ denotes the concatenation of $x$ and some other string $x'\in \Sigma^*$.
\item $B$-prefix denotes the first $B$ symbols of $x$. (Usually used when $\Sigma = \{0, 1\}$.)
\item $x^N$ the concatenation of $N$ number of string $x$.
\end{enumerate}

We use $U_n$ to denote the uniform distribution on $\{0,1\}^n$.

\subsection{Edit distance and longest common subsequence}

\begin{definition}[Edit distance] For any two strings $x, x'\in\Sigma^n$, the edit distance $ED(x,x')$ is the minimum number of edit operations (insertions and deletions) required to transform $x$ into $x'$.
\end{definition}

\begin{definition}[Longest Common Subsequence] For any strings $x, x'$ over $\Sigma$, the longest common subsequence of $x$ and $x'$ is the longest pair of subsequences of $x$ and $x'$ that are equal as strings. $LCS(x,x')$ denotes the length of the longest common subsequence between $x$ and $x'$.
\end{definition}

Note that $ED(x,x') = |x|+|x'|-2LCS(x,x')$.

\subsection{Almost k-wise independence}
\begin{definition}[$\eps$-almost $\kappa$-wise independence in max norm \cite{alon1992simple}]
Random variables $X_1, X_2, \ldots, X_n \in \{0,1\}^{n}$ are $\eps$-almost $\kappa$-wise independent in max norm if $\forall i_1, i_2, \ldots, i_{\kappa} \in [n]$, $\forall x \in \{0,1\}^{\kappa}$, $|\Pr[ X_{i_1} \circ X_{i_2} \circ \cdots \circ X_{i_{\kappa}} = x] - 2^{-\kappa} | \leq \eps.$

A function $g: \{0,1\}^{d} \rightarrow \{0,1\}^{n}$ is an $\eps$-almost $\kappa$-wise independence generator in max norm if $g(U) = Y = Y_{\kappa} \circ \cdots Y_{n}$ are $\eps$-almost $\kappa$-wise independent in max norm.
\end{definition}
In the following passage, unless specified, when we say $\eps$-almost $\kappa$-wise independence, we mean in max norm.

\begin{theorem}[$\eps$-almost $\kappa$-wise independence generator \cite{alon1992simple}]
\label{almostkwiseg}
There exists an explicit construction s.t. for every $n, \kappa \in \mathbb{N}$, $\eps > 0$, it computes an $\eps$-almost $\kappa$-wise independence generator $g: \{0,1\}^{d} \rightarrow \{0,1\}^n$, where $d = O(\log \frac{\kappa \log n }{\eps})$.

The construction is highly explicit in the sense that, $\forall i\in [n]$, the $i$-th output bit can be computed in time $\poly(\kappa, \log n, \frac{1}{\eps})$ given the seed and $i$.
\end{theorem}

\subsection{Pseudorandom generator}
\begin{definition}[Pesudorandom generator]
A generator $g:\{0,1\}^r\rightarrow \{0,1\}^n$ is a pseudorandom generator (PRG) against a function $f:\{0,1\}^n\rightarrow \{0,1\}$ with error $\eps$ if
\[\left|\Pr[f(U_n)=1]-\Pr[f(g(U_r))=1] \right|\leq \eps\]
where
 $r$ is called the seed length of $g$.

We also say $g$ $\eps$-fools function $f$. Similarly, $g$ $\eps$-fools a class of function $\mathcal{F}$ if $g$ fools all functions in $\mathcal{F}$.
\end{definition}

\begin{theorem}[PRG for CNF/DNFs \cite{de2010improved}]
\label{PRGforCNF}
There exists an explicit PRG $g$ s.t. for every $n, m\in \mathbb{N}, \eps > 0$,   every CNF/DNF $f$ with $n$ variables,   $m$ terms,
\[ |\Pr[f(U_n) = 1] - \Pr[f(g(n, m, \eps, U_r)) = 1]| \leq \eps\]
where  $|g(n, m, \eps, U_r)| = n$, $r =  O(\log n + \log^2(m/\eps) \cdot \log \log(m/\eps))$.
\end{theorem}

\subsection{Random walk on expander graphs}

\begin{definition}[$(n,d,\lambda)$-expander graphs]
If $G$ is an $n$-vertex $d$-regular graph and $\lambda(G)\leq \lambda$ for some number $\lambda<1$, then we say
that $G$ is an $(n, d, \lambda)$-expander graph. Here $\lambda(G)$ is defined as the second largest eigenvalue (in the absolute value) of the normalized adjacency matrix $A_G$ of $G$.

A family of graphs $\{G_n\}_{n\in \mathbb{N}}$ is a $(d, \lambda)$-expander graph family if there are some constants
$d\in\mathbb{N}$ and $\lambda<1$ such that for every $n$, $G_n$ is an $(n, d, \lambda)$-expander graph.
\end{definition}

\begin{theorem}[Random walk on expander graphs, \cite{alon1995derandomized}, \cite{hoory2006expander} Theorem 3.11]
\label{ExapanderRWHitting}
Let $A_0, \ldots, A_t$ be vertex sets of densities $\alpha_0,...,\alpha_t$ in an $(n, d, \lambda)$-expander
graph $G$. Let $X_0, \ldots, X_t$ be a random walk on $G$. Then
\[
\Pr\{ \forall i, X_i \in A_i\} \leq \prod_{i=0}^{t-1} (\sqrt{\alpha_i \alpha_{i+1}} + \lambda ).
\]
\end{theorem}

It is well known that some expander families have strongly explicit constructions.
\begin{theorem}[Explicit expander family \cite{arora2009computational}]
\label{explicitexpander}
For every constant $\lambda \in (0,1)$ and some constant $d \in \mathbb{N} $ depending on $\lambda$, there exists a strongly explicit
 $( d, \lambda)$-expander family.
\end{theorem}

\subsection{Error correcting codes (ECC)}


An $(n, m ,d)$-code $C$ is an ECC (for hamming errors) with codeword length $n$, message length $m$. The hamming distance between every pair of codewords in $C$ is at least $d$.

Next we recall the definition of ECC for edit errors.

\begin{definition}
An ECC $C\subseteq \{0,1\}^n$ for edit errors with message length $m$ and codeword length $n$ consists of an encoding mapping $Enc:\{0,1\}^m\rightarrow \{0,1\}^n$ and a decoding mapping $Dec:\{0,1\}^*\rightarrow \{0,1\}^m\cup\{Fail\}$. The code can correct $k$ edit errors if for every $y$, s. t. $ED(y,Enc(x))\leq k$, we have $Dec(y) = x$. The rate of the code is defined as $\frac{m}{n}$.

\end{definition}

An ECC family is explicit (or has an explicit construction) if
both encoding and decoding can be done in polynomial time.


We will utilize linear algebraic geometry codes to compute the redundancy of the hamming error case.
\begin{theorem}[\cite{hoholdt1998algebraic}]
\label{agcode}

There exists an explicit algebraic geometry ECC family  $ \{ (n, m, d)_q\mbox{-code } C \mid n, m \in \mathbb{N}, m \leq n, d = n-m-O(1), q = \poly(\frac{n}{d})\}$ with polynomial-time decoding when the number of errors is less than half of the distance.

Moreover, $\forall n, m \in \mathbb{N}$, for every message $x\in \mathbb{F}_q^{m}$, the codeword is $  x\circ z$ for some redundancy $z \in \mathbb{F}_q^{n-m}$.

\end{theorem}

To construct ECC from document exchange protocol, we need to use a previous result about asymptotically good binary ECC for edit errors given by  Schulman and Zuckerman \cite{796406}.

\begin{theorem}[\cite{796406}]
\label{asympGoodECCforInsdel}
There exists an explicit binary ECC family in which a code  with codeword length $n$, message length $m=\Omega(n)$,  can correct up to $k=\Omega(n)$ edit errors.

\end{theorem}

\section{Deterministic protocol for document exchange}
\label{sec:determdocexc}
We derandomize the IMS protocol given by Irmak et. al. \cite{irmak2005improved}, by first constructing $\eps$-self-matching hash functions and then use them to give a deterministic protocol.






\subsection{$\eps$-self-matching hash functions}
The following describes a matching property between strings under some given hash functions.

\begin{definition}\label{matchdef}
For every $n, n', t,  p, q  \in \mathbb{N}, q\leq p$, any hash functions $h_1, h_2, \ldots, h_{n'}$ where for every $i\in [n'], h_i: \{0,1\}^{p  } \rightarrow \{0,1\}^{q  }$, given two strings $x'\in (\{0,1\}^{p})^{n'}$ and $y \in \{0,1\}^{n}$, a (monotone) matching of size $t$ between $x', y$ under hash functions $h_1, \ldots, h_{n'}$  is a sequence of pairs of indices $w = ((i_1, j_1), (i_2, j_2), \ldots, (i_t, j_t)) \in ([n'] \times [n])^t$ s.t. 
$1\leq i_1 < i_2 < \cdots < i_t\leq n'$,   $j_{1}+p-1 < j_{2}, \ldots, j_{t-1}+p-1 < j_{t}, j_{t}+p-1 \leq n$ and $ \forall l\in [t], h_{i_l}(x'[i_l]) = h_{i_l}(y[j_l,  j_l+p-1] )$.

For $l\in [t]$, if $x'[i_l] = y[j_l,  j_l+p-1] $, we say $(i_l, j_l)$ is a good pair, otherwise we say it is a bad pair. We say $w$ is a correct matching if all pairs in $w$ are good. We say $w$ is a completely wrong matching is all pairs in $w$ are bad.

If parsing $x'$ to be binary we get $x$ which is equal to $y$, then $w$ is called a self-matching of $x'$ (or $x$).

For the match $w$ between $x', y$ under hash functions $h_1, \ldots, h_n$, we simply say a match $w$ if $x', y$ and $h_1, \ldots, h_{n'}$ are clear in the context.
\end{definition}

The next lemma shows that the maximum matching between two strings under some hash functions can be computed in polynomial time, using dynamic programming.
\begin{lemma}
\label{dpformatch}
There is an algorithm s.t. 
for every $n, n',  p, q  \in \mathbb{N}, q\leq p$, any hash functions $h_1, h_2, \ldots, h_{n'}$ where for every $ i\in [n']$, $h_i: \{0,1\}^{p  } \rightarrow \{0,1\}^{q  }$,  every  $x'\in (\{0,1\}^{p})^{n'}$ and $y \in \{0,1\}^{n}$, given $ h_1(x'[1])$, $\ldots$, $h_{n'}(x'[n'])$, $y$, $n$, $n'$, $p$, $q$, it can compute the maximum matching between $x', y$ under hash functions $h_1, \ldots, h_{n'}$  in time $O(n^2 (t_h + \log n))$, if for every $ i\in [n']$, $h_i$ can be computed in time $t_h$. (Note that $x$ is not necessary to be part of the input).
\end{lemma}

\begin{proof}

We present a dynamic programming  to compute the maximum matching.

For every $ j'\in [n'], j\in [n]$, let $f(j', j)$ be the size of the maximum matching between $x'[1, j']$ and $ y [1, j] $ under $h_1, \ldots, h_{j'}$. We compute $f$ as follows,
\[
f(j', j) =
\left \{
  \begin{array}{ll}
  \max(f(j'-1, j-p) + 1, f(j'-1, j), f(j', j-1) ),  &  \mbox{ if } h_{j'}(x'[j']) = h_{j'}(y[j-p+1, j]); \\
  \max(f(j'-1, j), f(j', j-1)), &  \mbox{ if } h_{j'}(x'[j']) \neq h_{j'}(y[j-p+1, j]).
  \end{array}
\right.
\]

Although $f$ only computes the size of the maximum matching, we can record the corresponding matching every time when we compute $f(j',j), j'\in [n'], j\in [n]$. So finally we can get the maximum matching after computing $f(n', n)$.

We need to compute $f(j', j), j'\in [n'], j\in [n]$ one by one and $n n' = O(n^2)$.  Every time we compute an $f(j', j)$, we need to compute a constant number of evaluations of the hash functions and append a pair of indices to some previous records to create the current record of the maximum matching. That takes $O(\log n + t_h)$. So the overall time complexity is $O(n^2(\log n+ t_h))$.
\end{proof}

Next we show that a matching between two strings with edit distance $k$, induces a self-matching in one of the strings, where the number of bad pairs decreases by at most $k$.
\begin{lemma}
\label{matchToSelfMatch}
For every $n, n', k, m, p, q  \in \mathbb{N}, k\leq n', q\leq p $, any hash functions $h_1, h_2, \ldots, h_{n'}$ where $\forall i\in [n'], h_i: \{0,1\}^{p  } \rightarrow \{0,1\}^{q  }$, given two sequences $x\in  \{0,1\}^{pn'}$ and $y \in \{0,1\}^{n}$ s.t. $\ED(x, y) \leq k$, parsing $x$ to be $x'\in (\{0,1\}^p)^{n'}$, if there exists a matching $w$ between $x'$ and $ y$ under $h_1,\ldots, h_{n'}$ having  at least $m$ bad pairs, then
there is a  self-matching $w'$ of $x$ with size $|w'| \geq |w| -k$, having  at least $m - k$ bad pairs.
%

\end{lemma}

\begin{proof}

Assume w.l.o.g. the $m$ bad pairs in the matching are the pairs $((j'_1, j_1), \ldots, (j'_{m}, j_{m} ) )$.

As the number of edit errors is upper bounded by $ k $, to edit $y$ back to $x$, we only need to do insertions and deletions on at most $k$ of substrings $y[j_{1}, j_1+p-1], y[j_2, j_2+p-1], \ldots, y[j_{m}, j_{m}+p-1]$. The substrings left unmodified induce a self-matching of $x$, which is a subsequence of $w$. Note that only at most $ k$ entries of $w$ are excluded. So the number  of bad pairs in the matching is at least $m - k$ and the size of $w'$ is at least $|w| - k$.

%

\end{proof}

The following property shows that a matching between two long intervals induces two shorter intervals having a matching s.t. the ratio between the matching size and the total interval length is maintained.
\begin{lemma}
\label{largeIntvToSmallOne}

For every $n, n', t,  p, q  \in \mathbb{N}, t\leq n, q\leq p $, any hash functions $h_1, h_2, \ldots, h_{n'}$ where $\forall i\in [n], h_i: \{0,1\}^{p  } \rightarrow \{0,1\}^{q  }$, given two sequences $x' \in (\{0,1\}^p)^{n'}$ and $y \in \{0,1\}^{n}$, if there is a matching $w$ between $x$ and $ y$ under $h_1,\ldots, h_{n'}$ having  size $  t $, then for every $ t^* \leq t$,
there is a matching $w^*$ between $x'_0, y_0$ having size   $ t^*  $, where $x'_0$ is an interval of $x'$, $y_0$ is an interval of $y$, and $ |x_0| + |y_0|/p  \leq \frac{2t^*}{t}(n'+n/p)$. Here $w^*$ is a subsequence of $w$.
\end{lemma}

\begin{proof}

We consider the following recursive procedure.

At the beginning (the first round), let $w_1 = w$. We know that  $n'+n/p \leq \frac{t}{t}(n'+n/p) $.

For the $i$-th round, assume we have a matching $w_i$ between $x'[j'_{ 1}, j'_{ 2}]$, $y[j_{ 1}, j_{ 2}]$ with $|w_i | = t_i$ and $  |x'[j'_{ 1}, j'_{ 2}]+y[j_{ 1}, j_{ 2}]|  \leq \frac{t_i}{t}(n'+n/p)  $. Let $w_i = (\rho'_{1}, \rho_2), \ldots, (\rho'_{t_i}, \rho_{t_i})$.

We pick $l = \lfloor t_i /2 \rfloor$. Consider the following two sequences $w_{i}[1, l]$ and $w_{i}[l+1, t_i]$.  Here $w_{i}[1, l]$ is a matching between $ x'_1 = x'[j'_{ 1}, \rho'_l]$ and $y_1 = y[j_{ 1}, \rho_l + p-1] $. Also $w_{i}[l+1, t_i]$ is a matching between $x'_2 = x'[\rho'_{l+1}, j'_{ 2}]$ and $ y_2 = y[ \rho_{l }+p, j_{ 2}] $.


Note that either $ (|x'_1| + |y_1|/p ) \leq \frac{l}{t}(n'+ n/p)$ or $ (|x'_2| + |y_2|/p ) \leq \frac{t_i - l}{t}(n'+ n/p) $.
Because if not,  then
$$ \left |x'[j'_{ 1}, j'_{ 2}] \right | + \left |y[j_{ 1}, j_{ 2}] \right | = (|x'_1| + |y_1|/p )+  (|x'_2| + |y_2|/p ) > \frac{t_i}{t}(n'+n/p),$$
which contradicts the assumption that $  |x'[j'_{ 1}, j'_{ 2}]+y[j_{ 1}, j_{ 2}]|  \leq \frac{t_i}{t}(n'+n/p)  $. Thus we can pick one sequence as $ w_{i+1}$.

We can go on doing this until   the $i^*$-th round in which there is a  matching $w_{i^*}$ of size   $t_{i^*}\in [t^*, 2t^*)$, whose corresponding pair of substrings are $x'_0$ and $y_0$. We know that $ |x'_0| + |y_0|/p  \leq \frac{t_{i^*}}{t}(n+n'/p) $. We pick $t^*$ pairs in $w_{i^*}$ and this gives a matching between $x'_0$ and $y_0$ s.t. $ |x'_0| + |y_0|/p  \leq \frac{2t^*}{t}(n'+n/p)$.

\end{proof}

We now describe an explicit construction for a family of sequences of $\eps$-self-matching hash functions.
\begin{theorem}
\label{selfmatchhash}
There exists an algorithm  which, on input $n,    p, q \in \mathbb{N}$, $\eps \in (0,1)$, $n\geq p \geq q$, $ q = \Theta(\log \frac{1}{\eps})  $, $x\in\{0,1\}^n$, outputs a description of $\eps $-self-matching functions $h_1, \ldots, h_{n' }:\{0,1\}^{p} \rightarrow \{0,1\}^q$, in time $\poly(n)$, where  the description length is $O(\log n)$ and $n' = \frac{n}{p}$.

Also there is an algorithm which, given  the same $ n, p, q\in \mathbb{N}$, the description of  $h_1, \ldots, h_{n'}$, $i\in [n']$ and any $ a\in \{0,1\}^p$, can output $h_{i}(a)$ in time $\poly(n)$.

\end{theorem}

To prove this theorem we consider the following construction.

\begin{construction}
\label{selfmatchhashconstruct}
Let $n,  p, q \in \mathbb{N}$, $\eps \in (0,1)$, $n \geq p \geq q$, $ q = \Theta(\log \frac{1}{\eps}) $, $x\in\{0,1\}^n$ be given parameters.

\begin{enumerate}

\item Divide $x$ into consecutive blocks to get $x' \in (\{0,1\}^{p})^{n'}$,  $n' = n/p $;

\item
Let $g:\{0,1\}^{d} \rightarrow \{0,1\}^{q n' 2^{p}}$ be the $\eps_g$-almost $\kappa$-wise independence generator from Theorem \ref{almostkwiseg}, where $\kappa = O(\eps n q)$,  $\eps_g = 1/\poly(n), d = O(\log n)$;

\item View the output of $g$ as in a two dimension array $(\{0,1\}^{q})^{ [n'] \times \{0,1\}^{p}}$;

\item 
 
 Exhaustively search $u \in \{0,1\}^{d}$ s.t.

\hypertarget{starchekpt}{($\star$)} For every consecutive $t_1 \leq \frac{4m}{p\eps}$-blocks $x'_0 \in (\{0,1\}^{p})^{t_1}$ of $(x'[1], \ldots, x'[n'])$ and every substring $x_0$ of $x$  having length $ t_2 \leq  \frac{4 m}{\eps}$, $x'_0$ and $x_0$ have no completely wrong matching of size $ m = \Theta(\frac{ \log n}{ q}) $ under functions  $h_i(\cdot) = g(u)[i][ \cdot], i\in [n']$, where $g(u)[i][ \cdot]$ is the corresponding entry in the two dimension array  $g(u )$;

\item return $u$ which is the description of the $\eps$-self-matching hash functions.
\end{enumerate}

(Evaluation of a hash function) Given  $u$, $i\in [n'], n, p, q, \eps$,   an input $v\in \{0,1\}^{p}$, one can compute $h_i(v) = g(u)[i][v]$.

\end{construction}

\begin{lemma}
\label{uexist}

There exists an $u$ such that the assertion \hyperlink{starchekpt}{($\star$)} holds.

\end{lemma}

\begin{proof}

We consider a uniform random string $u$ and show that with high probability $u$ satisfies assertion \hyperlink{starchekpt}{($\star$)}.

Fix a pair of substrings $x'[j'_1, j'_2]$, $ x[j_1, j_2]$, and a matching $w^* = ( (\rho'_1, \rho_1), \ldots, (\rho'_m, \rho_m) )$ between them of size $m = \Theta(\frac{ \log n}{ q})$. The probability
\begin{align*}
     & \Pr_{u} \{ w^* \mbox{ is a completely wrong matching } \}\\
\leq & \Pr_{u}\{ \forall l \in [m], h_{\rho'_{l}}(x'[\rho'_{l}]) = h_{\rho'_{\ell}} (x[\rho_{\ell}, \rho_{\ell} + p-1]) \} \\
\leq & \sum_{a \in (\{0,1\}^{p})^{m}} \Pr_{u}\{  \forall \ell\in [m], h_{\rho'_{\ell}}(x'[\rho'_{\ell}]) = h_{\rho'_{\ell}} (x[\rho_{\ell}, \rho_{\ell} +  p-1]) = a_{\ell} \}  \\
\leq & 2^{q m} ( (\frac{1}{2^{q }})^{2m} + \eps_g) \\
\leq & \frac{1}{2^{q m}} + 2^{qm}\eps_g \\
= & \frac{1}{\poly(n)},
\end{align*}

as long as we take $\eps_g$ to be  a sufficiently small $1/\poly(n)$. The total number of pairs $x'[j'_1, j'_2]$, $ x[j_1, j_2]$ is at most $\poly(n)$. For each fixed pair there are at most  ${t_1 \choose m} {t_2 \choose m} \leq   (\Theta(\frac{1}{\eps}))^{\Theta(\frac{ \log n}{  q})} =   (\Theta(\frac{1}{\eps}))^{\Theta(\frac{ \log n}{  \log \frac{1}{\eps} }) } = \poly(n )$ number of different matchings of size $m$. Thus by the union bound, the probability that the assertion \hyperlink{starchekpt}{($\star$)} holds is at least $1- 1/\poly(n)$ if we choose the parameters appropriately.

\end{proof}

\begin{lemma}
\label{selfmatchhashconstructcorrect}
Construction \ref{selfmatchhashconstruct} gives a sequence of $\eps$-self-matching hash functions. 

\end{lemma}

\begin{proof} 

Let $w$ be a self-matching of $x$, having at least $ \eps n +1  $ bad pairs.  We pick all the wrong pairs in this matching to get
a completely wrong self-matching $\tilde{w}$.

By Lemma \ref{largeIntvToSmallOne},   there are two substrings, $x'[j'_1, j'_2]$ and $x[j_1, j_2]$, $ 1\leq j'_1 \leq j'_2 \leq [n'], 1 \leq j_1 \leq j_2 \leq [n], |x'[j'_1, j'_2]| + |x[j_1, j_2]|/p \leq \frac{2m}{\eps n+1}(n'+n/p)$, having a completely wrong matching of size at least $ m  $, which is a subsequence of $\tilde{w}$. Note that $ |x'[j'_1, j'_2]| \leq \frac{2m}{\eps n+1}(n'+n/p) =   \frac{4m n}{(\eps n+1)p} \leq  \frac{4m}{\eps p}$. Also $|x[j_1, j_2]| \leq \frac{2m}{\eps n+1}(n/p+n/p) p = \frac{4m n}{\eps n+1} \leq  \frac{4m  }{\eps  } $.  This contradicts \hyperlink{starchekpt}{($\star$)}.

\end{proof}

\begin{lemma}
\label{evltime}
The evaluation of a function in the sequence takes polynomial time.
\end{lemma}

\begin{proof}
For evaluation, by Lemma \ref{almostkwiseg}, given a seed and a position index, the corresponding bit in $g$'s output can be computed in time $ \poly(\kappa, \log (qn'2^{ p}), \frac{1}{\eps_g}) = \poly(n)$. The output of a function has $ q$ bits. One can compute the $q$ bits one by one and the total running time is still $\poly(n)$.

\end{proof}

\begin{lemma}
\label{selfmatchhashtime}
The algorithm runs in polynomial time.

\end{lemma}

\begin{proof}

Checking the assertion \hyperlink{starchekpt}{($\star$)} takes time $\poly(n)$. To see this, first note that there are $ \poly(n)$ pairs of $x'_0$ and $x_0$. Also for each pair, the total number of matchings of size $m$ is at most
\[ {t_1 \choose m} {t_2 \choose m} \leq {\frac{4m}{p\eps} \choose m} {\frac{4m}{ \eps} \choose m}  \leq (\frac{4 e}{ \eps})^{2m}  = (\frac{4 e}{ \eps} )^{\Theta(\frac{ \log n}{ q})} =   (\frac{4 e}{ \eps} )^{\Theta(\frac{ \log n}{  \log \frac{1}{\eps} }) } = \poly(n ). \]
For a specific matching, Alice can first compute the hash values in time $\poly(n)$ by Lemma \ref{evltime}. Alice can check whether it is a wrong matching and compute the size of the matching in time $\poly(n)$.

Note that there are $\poly(n)$ number of different seeds. So the algorithm of generating the description runs in polynomial time.

\end{proof}

\begin{proof}[Proof of Theorem \ref{selfmatchhash}]

We use Construction \ref{selfmatchhashconstruct}, the theorem
 directly follows from Lemma \ref{uexist}, \ref{selfmatchhashconstructcorrect}, \ref{selfmatchhashtime}.

\end{proof}

\subsection{Deterministic protocol for document exchange}
Our deterministic protocol for document exchange is as follows.

\begin{construction}
\label{dIMS}
The protocol is for every input length $n\in \mathbb{N}$,  every  $k \leq \alpha n$ number of edit errors where $\alpha$ is a constant. For the case $k > \alpha n$, we simply let Alice send her input string.

Both Alice's and Bob's algorithms have $L = O(\log \frac{n}{k})$ levels.

Alice: On input $x \in \{0,1\}^n$;
\begin{enumerate}[label*=\arabic*.]

\item We set up the following parameters;
\begin{itemize}

\item For every $i\in [L]$, in the $i$-th level,
\begin{itemize}
\item The block size is $b_i =  \frac{n}{3\cdot 2^{i} k } $, i.e., in each level we divide a block in the previous level evenly into two blocks. We choose $L$ properly s.t. $b_L = O(\log \frac{n}{k})$;

\item The number of blocks $l_i = n/b_i$;
\end{itemize}

\end{itemize}

\item For the $i$-th level,
\begin{enumerate}[label*=\arabic*.]

\item Divide $x$ into consecutive blocks to get $x'\in (\{0,1\}^{b_i})^{l_i}$;

\item Construct a sequence of $\eps = \frac{k}{n}$-self-matching hash functions $h_1,\ldots, h_{l_i}:\{0, 1\}^{b_i} \rightarrow \{0,1\}^{b^*}$ for $x$ by Theorem \ref{selfmatchhash}, with $b^* = O(\log \frac{n}{k})$. Let  the description of the hash functions be $u[i] \in \{0 ,1\}^{O(\log n)}$ by Theorem \ref{selfmatchhash};

%
%

\item Compute $v[i] = ( h_1(x'[1]), h_2(x'[2]), \ldots, h_{l_i}(x'[l_i]))$;


\item Compute the redundancy $ z[i] \in (\{0,1\}^{b^*})^{\Theta(k)} $ for $v[i]$ by Theorem \ref{agcode}, where the code has distance $14k$;
\end{enumerate}


\item Compute the redundancy $z_{\mathsf{final}} \in (\{0,1\}^{b_L})^{\Theta(k)} $ for the blocks of the $L$-th level  by Theorem \ref{agcode}, where the code has distance $8k$;

\item Send $u = (u[1], u[2], \ldots, u[L])$, $z = (z[1], z[2], \ldots, z[L])$, $v[1]$, $z_{\mathsf{final}}$.

\end{enumerate}

Bob: On input $y \in \{0,1\}^{O(n)}$ and received $u, z$, $v[1]$, $z_{\mathsf{final}}$;

\begin{enumerate}[label*=\arabic*.]

\item Create $\tilde{x} \in \{0, 1, *\}^{n}$ (i.e. his current version of Alice's $x$), initiating it to be $\{*, *, \ldots, *\}$;

\item For the $i$-th level where $1 \leq i \leq L-1$,
\begin{enumerate}[label*=\arabic*.]

\item Divide $\tilde{x}$ into consecutive blocks to get $\tilde{x}' \in (\{0,1\}^{b_i})^{l_i}$;

\item Apply the decoding of Theorem \ref{agcode} on $ h_1(\tilde{x}'[1])\circ h_2(\tilde{x}'[2])\circ \ldots \circ h_{l_i}(\tilde{x}'[l_i]) \circ z_i$ to get the sequence of hash values $ v[i] = ( h_1(x[1])$, $h_2(x[2])$, $\ldots$, $h_{l_i}(x[l_i]))$.  Note that $v[1]$ is received directly, thus Bob does not need to compute it;

\item Compute $w = ((\rho'_1, \rho_1), \ldots, (\rho'_{|w|}, \rho_{|w|})) \in ([l_i] \times [|y|])^{|w|}$ which is the maximum matching between $x'$ and $y$ under $h_1, \ldots, h_{l_i}$, using $v[i]$, by Lemma \ref{dpformatch};

\item Evaluate $\tilde{x}$ according to the matching, i.e. let  $\tilde{x}'[\rho'_j] = y[\rho_j, \rho_j + b_i - 1]$;

\end{enumerate}

\item In the $L$'th level, apply the decoding of Theorem \ref{agcode} on the blocks of $\tilde{x}$ and $ z_{\mathsf{final}}$ to get $x$;

\item Return $x$.

\end{enumerate}

\end{construction}

\begin{lemma}
\label{dIMS:timec}
Both Alice's and Bob's algorithms
are in polynomial time.
\end{lemma}

\begin{proof}
We first consider Alice's algorithm. For the $i$-th level, $i\in [L]$, dividing $x$ into blocks takes time $O(n)$. Computing the description and doing evaluation of $\eps$-self-matching hash functions takes time $\poly(n)$ by Theorem \ref{selfmatchhash}.

The redundancy $z[i], i \in [L]$ and $z_{\mathsf{final}}$ can be computed in polynomial time by Theorem \ref{agcode}.

Thus the time complexity for Alice is $\poly(n)$.

Next we consider Bob's algorithm. Creating $\tilde{x}'$ at the beginning of each level takes $O(n)$ time. Decoding of $  h_1(\tilde{x}'[1])\circ h_2(\tilde{x}'[2])\circ \ldots \circ h_{l_i}(\tilde{x}'[l_i]) \circ z[i]$ takes $\poly(n)$ time by Theorem \ref{agcode}. By Lemma \ref{dpformatch}, computing the maximum matching between $x'$ and $y$ under $h_1, \ldots, h_{l_i}$ takes time $O(n^2 (t_h + \log n))$ where $t_h$ is the time complexity of evaluating any one of the hash functions. By Theorem \ref{almostkwiseg}, $t_h$ is $\poly(n)$, so computing the maximum matching takes $\poly(n)$. Decoding on the last level of $\tilde{x}$'s blocks and  $z_{\mathsf{final}}$ also takes polynomial time by Theorem \ref{agcode}.

So the overall running time for Bob is also $\poly(n)$.
\end{proof}

\begin{lemma}
\label{dIMS:communicationc}
The communication complexity  
is $O(k \log^2 \frac{n}{k})$.
\end{lemma}

\begin{proof}

For every $i\in [L]$, $|u[i]| = O(\log n)$ by Theorem \ref{selfmatchhash}. Thus the total number of bits of $u$ is $L \cdot O(\log n) = O( \log \frac{n}{k} \log n)$.

Also for every $i\in [L]$, by Theorem \ref{agcode} the number of bits in $z_i$ is $O(k \log \frac{n}{k})$.  So the total number of bits of $z$ is $O(k \log^2 \frac{n}{k})$.

Again by Theorem \ref{agcode}, $z_{\mathsf{final}}$  has  $O(k \log \frac{n}{k}) $ bits.

Note that $v[1]$ has $O(k \log \frac{n}{k}) $ bits, since there are $l_1 = O(k)$ blocks in the first level and the length of each hash value is $ O(\log \frac{n}{k})  $.

Thus the total communication complexity is $O( \log \frac{n}{k} \log n + k \log^2 \frac{n}{k}) = O( k \log^2 \frac{n}{k})$, since $k \log \frac{n}{k} \geq \log n$ when $k \leq \alpha n$.

\end{proof}

Next we show the correctness of the construction. We first establish a series of lemmas.
\begin{lemma}
\label{missbond}

For any $i \leq L-1$, if $ v[i]$ is correctly recovered, then in the $i$-th level the number of bad pairs in $w$ is at most $2k$.

\end{lemma}

\begin{proof}

We prove by contradiction.

Suppose there are more than $2k$ bad pairs in $w$.
By Lemma \ref{matchToSelfMatch}, there is a self-matching having at least $k$ bad pairs.\ By picking all the wrong pairs in this matching, we get
a completely wrong self-matching $\tilde{w}$ having size at least $k = \eps n$. This is a contradiction to the fact that $h_1, \ldots, h_{l_i}$ is a sequence of $\eps$-self-matching hash functions.

\end{proof}

Next we show that $w$ is large enough so that in each level Bob can recover many blocks of $x$ correctly.

\begin{lemma}
\label{lenofw}

For any $i \leq L-1$, in the $i$-th level, we have $|w| \geq l_i-k$.
\end{lemma}

\begin{proof}
The $k$ edits which the adversary makes on $x$ can change at most $k$ blocks of $x'$. The remaining unchanged blocks induce a matching between $x' $ and $y$ of size at least $l_i-k$. Since $w$ is the maximum matching between $x'$ and $y$, $|w| \geq l_i-k$.

\end{proof}

\begin{lemma}
\label{dIMS:correct}
Bob computes $x$ correctly.
\end{lemma}

\begin{proof}
Alice can do every step in her algorithm correctly due to Theorem \ref{selfmatchhash} and \ref{agcode} since she only constructs $\eps$-self-matching hash functions, doing evaluation of these functions and computing redundancies of some sequences.

So the remaining is to show that Bob can compute $x$ correctly once he receives Alice's message.

We first show that for every level $i$, Bob can recover $v[i]$ correctly, by indcution.

For the first level, Bob can do it because $ v[1] $ is sent directly to him from Alice.

For level $i = 2, \ldots, L-1$, assume Bob gets $v[i-1]$ correctly. By Lemma \ref{missbond}, the matching $w$ has at most $2k$ bad pairs. By Lemma \ref{lenofw}, $|w| \geq l_i-k$. Thus $ |w| $ gives at least $l_i-k-2k = l_i - 3k$ correctly matched pairs of blocks. So according to $w$, Bob can recover at least $l_i-3k$ blocks of $x$ correctly. Thus in the $i$-th level, there are at most $3k \times 2 = 6k$ wrong blocks in $\tilde{x}$. So $h_1(\tilde{x}_1)\circ \ldots \circ h_{l_i}(\tilde{x}_{l_i}) \circ z[i]$ is a word in $(\{0,1\}^{b^*})^{l_i+\Theta(k)}$ having distance  at most $6k$ to a codeword of an $(l_i+ \Theta(k), l_i, \Theta(k))$-code. Let the distance of the code $\Theta(k)$ be at least $14k$  s.t. by Theorem \ref{agcode} the decoding algorithm can compute $v[i]$ correctly.

Note that for the last level, there are at least $n-3k$ correctly matched pairs of blocks. So there are at most $3k$ wrong blocks in $\tilde{x}$. The redundancy length $|z_{\mathsf{final}}| = \Theta(k)$. So $ \tilde{x} \circ z_{\mathsf{final}} $ has hamming distance at most $3k$ from a codeword $x \circ z_{\mathsf{final}}$ of an $(n+ \Theta(k) ,n, \Theta(k) )$-code with distance at least $8k$. Thus the decoding algorithm of Theorem \ref{agcode} can compute the message $x$ correctly.

\end{proof}

\begin{theorem}

\label{deterdocexc}
There exists a deterministic protocol for document exchange, having communication complexity (redundancy) $O(k \log^2 \frac{n}{k}) $, time complexity $\poly(n)$, where $n$ is the input size and $k$
is the edit distance  upper bound.

\end{theorem}

\begin{proof}

Construction \ref{dIMS} gives the deterministic protocol for document exchange. The communication complexity is $O(k \log^2 \frac{n}{k}) $ by Lemma \ref{dIMS:communicationc}. The time complexity is $\poly(n)$ by Lemma \ref{dIMS:timec}. The correctness is proved by Lemma \ref{dIMS:correct}.

\end{proof}

\section{Document exchange for a uniform random string}\label{sec:randdocexc}

In this section we prove the following theorem.

\begin{theorem} \label{randdocexmaintheorem}
	There exists a deterministic document exchange protocol for a uniformly random string with success probability $1 - 1/\poly(n)$ and
	redundancy size $O(k \log n)$.
\end{theorem}

We will modify the protocol to obtain an error correcting code in Section~\ref{sec:InsdelCode}.

\subsection{String properties}
\begin{definition}[$p$-split point\cite{8022906}]
	For a string $p \in \{0, 1\}^s$ and $x \in \{0, 1\}^n$, a $p$-split point of $x$ is an index $1 \le i \le n - s + 1$ such that
	$x[i, i+s) = p$.
\end{definition}

\begin{definition}[next p-split point]\label{nextsplitpoint}
	Let $p$ and $x$ be two strings,
	and $i$ be a $p$-split point of $x$.
	Define the \emph{next $p$-split point} of $i$ to be the smallest $j$ such that $j$ is a $p$-split point and $j > i$.
	If such $j$ does not exist, we define the \emph{next $p$-split point} of $i$ to be $(n+1)$.
\end{definition}

We will use the following properties of a uniform random string.
\begin{theorem} \label{Bdistinctandinterval}
	For a uniform random string $x \in \{0, 1\}^n$, let $p = 1 \circ 0^{s-1}$ be a string of length $s$.
	There exist three integers $B_1 = O(s 2^s \log n), B_2 = O(2^s \log n)$ and $B = O(\log n)$ such that 
	the following properties hold with probability $1 - 1/\poly(n)$.
	
	\begin{description}
		\item [\namedlabel{uniformproperty1}{Property 1}] Any interval of $x$ with length $B_1$ contains a $p$-split point.
		
		\item [\namedlabel{uniformproperty2}{Property 2}] Any interval of $x$ with length  $B_2$ starting at a $p$-spit point contains a $p$-split point $i$ such that its next $p$-split point $j$ satisfies $j - i > 2^s / 2$.
		
		\item [\namedlabel{Bdistinct}{$B$-distinct}] Every two substrings of length $B$ at different positions of $x$ are distinct.
	\end{description}
\end{theorem}

\begin{proof}
	For \ref{uniformproperty1}, let $B_1 = 2 s 2^s \log n$.
	For any fixed interval $I$ of length $B_1$, we divide $I$ into small intervals of length $s$.
	Then we have at least $2\cdot 2^s \log n$ small intervals.
	As the strings in each small interval are independent, the probability that $p$ does not appear in $I$ is at most 
	\begin{equation}
	\left(1 - \frac{1}{2^s}\right)^{2 \cdot 2^s \log n} < \left(\frac{1}{e}\right)^{2 \log n} < \frac{1}{n^2}
	\end{equation}
	Applying the union bound,
	the probability that such an interval $I$ exists is at most $n / n^2 = 1/n$.
	
	For \ref{uniformproperty2}, let $B_2 = 2^s \log n$. If there is a fixed interval $I = [i, i + B_2)$ such that $i$ is a $p$-split point,
	and $I$ violates the second property, i.e., the $p$-split points $i_1 < i_2 < i_3... $ in $I$ satisfy $i_{t+1} - i_t \le 2^s/2$ for $t = 1, 2, \dots$.
	Denote the random variables $x_t = i_{t+1} - i_t, t = 1, 2, \dots$. Applying the union bound, $\Pr[x_t \le 2^s / 2] \le (2^s / 2) /2^s = 1/2$.
	As $x_t$ are independent, and we have at least $B_2 / (2^s/2) = 2 \log n$ such random variables $x_t$,
	we have
\[\Pr[x_t \le 2^s / 2, \forall t \le 2\log n] \le (\Pr[x_t \le 2^s / 2])^{2 \log n} = (1/2)^{2 \log n} < \frac{1}{n^2}\]
	Applying the union bound, the probability that such an interval $I$ exists is at most $n / n^2 = 1/n$.
	
	For \ref{Bdistinct}, let $B = 3 \log n$. Consider any two distinct fixed indexes $1 \le i \le n - B + 1$ and $1 \le j \le n - B + 1$.
	W.L.O.G, we can assume $i < j$. There are two cases: $j - i \ge B$ and $j - i < B$.
	If $j - i \ge B$, then $\Pr[x[i, i+B) = x[j, j+B)] = \sum_{r\in\{0, 1\}^B} \Pr[x[j, j+B) = r \mid x[i, i+B) = r ] \Pr[x[i, i+B) = r] = 1/2^B$. If $j - i < B$, then $x[j, j+B)$ is equal to the $B$-prefix of $x[i, j) \circ x[i, j) \circ x[i, j) \dots$,
	which is completely determined by $x[i, j)$.
	So we have

\begin{align*}
&\Pr[x[i, i+B) = x[j, j+B)]\\
 = &\sum_{r \in \{0, 1\}^{j - i}} \Pr[x[j, j+B)
 = (r \circ r \circ r \dots )[1, B] \mid x[i, j) = r] \Pr[x[i, j) = r] \\
 = &1/2^B \sum_{r} \Pr[x[i, j) = r] \\
 = &1/2^B
\end{align*}
	By the union bound, the probability that $B$-distinct property doesn't hold is at most $n^2 / 2^B = n^2 / n^3 = 1/n$.
	
	By the union bound, the probability that all properties are satisfied is at least $1 - 3/n$.
\end{proof}

\subsection{$\eps$-synchronization hash function}
In this subsection,
we introduce the notation of
$\eps$-synchronization hash function.

\begin{definition}[$\eps$-synchronization hash function] \label{epssyncfuncdef}
	For any $0 < \eps < 1$ and $T, n' \in \N, T \geq B$, let $x$ be a \ref{Bdistinct} string of length $n = n' T$, and $x_T = (x[1, T], x[T+1, 2T], \dots, x[(n'-1)T+1, n])$ be a partition of $x$.
	Let $R$ be an integer, and $\Phi = (\Phi[1], \Phi[2], \dots, \Phi[n'])$ be a sequence of functions,
	where $\Phi[t]$ maps the B-prefix of $x_T[t]$ to $\{0, 1\}^R$ for any $t \in [n']$.
	
	If $y$ is a string of length $m$, $0 \le l_1 < r_1 \le n'$
	and $0 \le l_2 < r_2 \le m$ be some indexes,
	let $\mathsf{MATCH}_{\Phi}(x_T(l_1, r_1], y(l_2, r_2])$ denote the size of the maximum matching between $x_T(l_1, r_1]$
	and $y(l_2, \min(r_2+T, m+1))$ under $\Phi(l_1, r_1]$ (See Definition~\ref{matchdef}).
	
	$\Phi$ is a sequence of  $\eps$-synchronization hash function with respect to $x$, if it satisfies the following properties
	\begin{itemize}
		\item For any three integers $i, j, t$ where $i < Tj$ and $j < t$, denote $l_1 = t - j$ and $l_2 = Tj - i$.
		\begin{equation}
		\mathsf{MATCH}_{\Phi}(x_T(j, t], x(i, Tj]) < \eps \left(
		l_1 + \frac{l_2}{T}
		\right) \label{eq-sync-1}
		\end{equation}
		
		\item For any three integers $i, j, t$ where $i < j$ and $t > T(j-1)+1$, denote $l_1 = j - i$ and $l_2 = t - T(j-1) - 1$.
		\begin{equation}
		\mathsf{MATCH}_{\Phi}(x_T(i, j], x(T(j-1)+1, t]) < \eps \left(l_1 + \frac{l_2}{T}
		\right)
		\label{eq-sync-2}
		\end{equation}
	\end{itemize}
\end{definition}

%

\begin{lemma}\label{philemma}
Let $0<\eps <1$ be a constant.
Let $x\in \{0,1\}^n$ be \ref{Bdistinct} and $x_T$ be $x$ parsed into blocks of length $T$ with $T \geq B$. There exists $ \Phi = (\phi[1], \ldots, \phi[m])$ where $m = n/T, \forall t\in [m], \phi[t]: \{0,1\}^{B} \rightarrow \{0,1\}^R, R = \Theta(\log^{0.5} n + \log T)$ s.t. for any $0\leq i_1 < i_2 \leq m, 0\leq j_1 < j_2 \leq n, l_1 = i_2-i_1, l_2 = j_2-j_1, l_1 + l_2/T \geq   \log^{0.6} n  $, when the overlap of $ x(i_1 T, i_2 T ] $ and $x(j_1, \min(j_2 + T - 1, m)]$ has less than $T$ bits,
\[ \mathsf{MATCH}_{\Phi}(x_T(i_1, i_2], x(j_1, j_2]) < \eps \left(
		l_1 + \frac{l_2}{T}
		\right).  \]

There is an algorithm s.t. given $n\in \mathbb{N}$ and $x$, it can compute a description of $\Phi$ in polynomial time, where the length of the description is $O(\log n)$. Also, there is an algorithm s.t. given the description of $\Phi$, $i\in [m]$ and an input in $\{0,1\}^B$ can compute the corresponding output of $\phi_i$ in polynomial time.

\end{lemma}

\begin{proof}

Let $g:\{0,1\}^{d} \rightarrow \{0,1\}^{R m 2^B}$ be an $\eta $-almost $\kappa$-wise independence generator from Theorem \ref{almostkwiseg}, where $\eta  = 1/\poly(n), \kappa = O(\frac{\log n}{\eta } R)$, $d = O(\log n)$.
%

Let's fix $ \ell  = \Theta(\frac{\log n}{  R})$.
We view the output of $g$ as in $(\{0,1\}^R)^{[m] \times \{0,1\}^B}$. Then we exhaustively search a seed $u$ for $g$ s.t. for any length $l^*_1  $ interval $x_T(i^*_1, i^*_2]$ and any length $l^*_2 $ interval $x(j^*_1, j^*_2]$ where $l^*_1 + l^*_2/T \leq \frac{3 \ell}{\eps} $ and the overlap between the two intervals has less than $T$ bits, there does not exist a matching of size $\ell$ between the two intervals under  hash functions $\Phi(i^*_1 , i^*_2]$.
Finally, let $ \phi[t](\cdot) = g(u)[t][\cdot]$ for every $ t\in [m]$.
%

Now we show the correctness of the construction by contradiction. Suppose there are $x_T(i_1, i_2], x(j_1, j_2]$   s.t. $ \mathsf{MATCH}_{\Phi}(x_T(i_1, i_2], x(j_1, j_2]) \geq \eps  ( l_1 + \frac{l_2}{T} ) \geq \eps  \log^{0.6} n  \geq \ell$ and the overlap between $ x(i_1 T, i_2 T ] $ and $x(j_1, j_2]$ has length less than $T$. Then the match is actually between $x_T(i_1, i_2]$ and $x(j_1, \min(j_2 + T - 1, m)]$.
By Lemma \ref{largeIntvToSmallOne}, there exist intervals $x_T(i^*_1, i^*_2]$ with length $l^*_1  $ and $x(j^*_1, j^*_2]$ with length $l^*_2 $ s.t. $l^*_1 + l^*_2/T \leq \frac{2 \ell}{\eps(l_1 + \frac{l_2}{T})} (l_1 + \frac{l_2 + T-1}{T}) \leq \frac{3\ell}{\eps}$ and  there is a matching of size $\ell$ between the two intervals. Also note that by Lemma \ref{largeIntvToSmallOne},  $x_T(i^*_1, i^*_2]$ is a subinterval of $x_T(i_1, i_2]$ and $ x(j^*_1, j^*_2] $ is a subinterval of $x(j_1, \min(j_2 + T - 1, m)]$. So the overlap between $x_T(i^*_1, i^*_2]$ and $ x(j^*_1, j^*_2] $ has less than $T$ bits. This is a contradiction to our choice of $u$.

To finish the proof,
we need to show that there exists such a $u$. We show that a uniform random choice of $u$ satisfies the requirement with high probability. Fix $i^*_1, i^*_2, j^*_1, j^*_2$, where $ i^*_2- i^*_1 = l^*_1$ and $j^*_2 - j^*_1 = l^*_2$. Fix a sequence of pairs
$w= ((\rho_1, \rho'_1), \ldots, (\rho_{\ell}, \rho'_{\ell}))$
between $x_T(i_1^*, i_2^*], x(j_1^*, j_2^*]$,  the probability that $w$ is a matching under $\Phi$ is
\begin{equation}
\begin{split}
	 &\Pr \left [ \forall t^*\in [\ell], \phi[{t^*}]\left(x_T[\rho_{t^*}] \right) = \phi[t^*]\left( x[\rho'_{t^*}, \rho'_{t^*} + T-1] \right) \right] \\
\leq & \sum_{a_1, a_2, \ldots, a_{\ell}\in \{0,1\}^R} \Pr\left [ \forall t^*\in [\ell], \phi[t^*]\left(x_T[\rho_{t^*}]\right) = \phi[t^*]\left(x[\rho'_{t^*}, \rho'_{t^*} + T-1]\right) = a_{t^*}\right ] \\
\leq & 2^{R \ell} ( 2^{-2R\ell} + \eta  )\\
\leq & 1/\poly(n),
\end{split}
\end{equation}
as long as $\eta$ is small enough. Note that the overlap between  $x(i_1^*T, i_2^*T], x(j_1^*, j_2^*]$ has less than $T$ bits. So for every $t^* \in [n]$, the $B$-prefix of $  x_T[\rho_{t^*}]$ and $  x[\rho'_{t^*}, \rho'_{t^*} + T-1] $ are not equal.

There are at most $ {l^*_1 \choose \ell} {l_2^* \choose \ell} \leq (2e/\eps)^{\ell}(2e T/\eps)^{\ell}  = (2e/\eps)^{O(\frac{\log n}{  R})}(2e T/\eps)^{O(\frac{\log n}{  R})} = \poly(n)$ different matchings between  $x_T(i^*_1, i^*_2]$ and $x(j^*_1, j^*_2]$. The total number of different  $i^*_1, i^*_2, j^*_1, j^*_2$ is at most $\poly(n)$. As long as $\eta$ is small enough, by a union bound, a uniform random choice of $u$ satisfies the requirement with probability at least $1-1/\poly(n)$.

The time complexity of this procedure is $\poly(n)$. This is because in the exhaustive search we try $ O(n^2l^*_1 l^*_2) = \poly(n) $ pairs of intervals and for each pair we try $ {l^*_1 \choose \ell} {l_2^* \choose \ell} \leq (2e/\eps)^{\ell}(2e T/\eps)^{\ell} = \poly(n)$ different matchings. Moreover, we try at most $2^d = \poly(n)$ number of different seeds to find $u$.

\end{proof}

\begin{lemma}\label{thetalemma}
 	Let $S$ be a subset of $\{0, 1\}^B$ of size $m$.
	There exists a map $\theta$ from $\{0, 1\}^B$ to $\{0, 1\}^{R_{\theta}=O(\log m)}$ such that
	$\theta$ restricted on $S$ is injective and $\theta$ can be represented in $d  = O(\log B) + O(\log m)$ bits.
	In addition, there is an algorithm s.t. given $B, m, S$, it can find such a $\theta$ whose binary representation has the smallest lexicographical order in time $\poly(B, m)$; given $B, m, \theta$ and an input in $\{0,1\}^B$, the output of $\theta$ can be computed in time  $ \poly(B, m) $.

\end{lemma}

\begin{proof}

Let $g:\{0,1\}^{d} \rightarrow \{0,1\}^{R_{\theta} 2^B }$ be an $\eta$-almost $2R_{\theta}$-wise independence generator from Theorem \ref{almostkwiseg}, where $\eta = 1/\poly(m)$.
Regard the output of $g$ as a sequence in $(\{0,1\}^{R_{\theta}})^{\{0,1\}^B}$.
We exhaustively find $u\in \{0,1\}^d$ s.t. $\theta$ is injective restricting on $S$, where $ \theta(\cdot) = g(u)[\cdot] $.
It remains to show that there exists such a $u$. Assume $u$ is drawn from $U_d$. For each fixed pair of distinct elements $ a, b\in S $,
\[
	\Pr[\theta(a) = \theta(b)] = \sum_{v \in \{0,1\}^{R_{\theta}}}\Pr[\theta(a) = \theta(b) = v] \leq 2^{R_{\theta}}(2^{-2R_{\theta}} + \eta) \leq \frac{1}{m^3},
\]
as long as $\eta$ is small enough. Note that there are at most $m^2$ such pairs. By a union bound, with probability $1-1/m$, $\theta$ is injective restricted on $S$.

Since $g$ is highly explicit by Theorem \ref{almostkwiseg}, the evaluation of $\theta$ can be computed in time $\poly( B,  m )$.

Note that the time complexity for finding $\theta$ is $\poly(B, m)$ because the exhaustive search tries $O(2^d)$ different seeds. Checking whether $\theta$ is injective on $S$ takes time $O(m^2 ) \poly(B, m) = \poly(B, m)$.

\end{proof}

\begin{construction} \label{epssynchashfunctionconstruction}
For any $0 < \eps < 1$ and $T \in \N, T \geq B$, let $x$ be a \ref{Bdistinct} string of length $n = Tn'$, and
$x_T = (x[1, T], x[T+1, 2T], \dots, x[(n'-1)T+1, n])$ be a partition of $x$.
	
	For any $t \in [n']$,
	let $S_t$ be the set containing all the substrings $x[u, u+B)$ satisfying $|T(t-1) + 1 - u| < T\cdot \log^{0.6} n$ and $1 \le u \le n - B + 1$.
	Applying Lemma~\ref{thetalemma},
	we choose the $\theta[t]$ whose binary representation has the smallest lexicographical order.
	Let $\phi$
	be a series of functions generated in Lemma~\ref{philemma}. For any $t \in [n']$, let
	$\Phi[t] = (\phi[t], \theta[t])$, and define $\Phi[t](s) = (\phi[t](s), \theta[t](s))$ for all $s \in \{0, 1\}^B$.
\end{construction}

\begin{theorem}
	The function sequence $\Phi$ in Construction~\ref{epssynchashfunctionconstruction} is a sequence of $\eps$-synchronization hash functions  with respect to $x$.
\end{theorem}

\begin{proof}
	For three integers $i, j, k$ where $i < Tj$ and $j < k$, denote $l_1 = k - j$ and $l_2 = Tj - i$.
	
	If $l_1 + l_2 / T \ge \log^{0.6} n$, by Lemma~\ref{philemma},
	Equation~\ref{eq-sync-1} holds.
	As every matching under $\Phi$ must be a matching under $\phi$, we have
	\begin{equation}
	\mathsf{MATCH}_{\Phi}(x_T(j, k], x(i, Tj]) \le
	\mathsf{MATCH}_{\phi}(x_T(j, k], x(i, Tj]) \le \eps\left(l_1 + \frac{l_2}{T}\right).
	\end{equation}
	
	If $l_1 + l_2 / T < \log^{0.6} n$, the total length of the interval $(i, Tk]$ is $Tk - i = T l_1 + l_2 < T\log^{0.6} n$. From the construction of $(\theta[t])_{t \in [n']}$, $\theta[t], (j < t \le k)$ maps the substrings of length $B$ in $(i, Tk]$ into different values. Hence we have $\mathsf{MATCH}_{\Phi}(x_T(j, k], x(i, Tj]) = 0$.
	
	Thus Equation~\ref{eq-sync-1} holds for all possible $l_1$ and $l_2$. By symmetry Equation~\ref{eq-sync-2} also holds.
\end{proof}

\begin{lemma}\label{epshashlemma}
	For any $0 < \eps < 1$ and $T \in \N, T \geq B$, let $x$ be a \ref{Bdistinct} string of length $n = Tn'$, and
	$x_T = (x[1, T], x[T+1, 2T], \dots, x[(n'-1)T+1, n])$ be a partition of $x$.
	Let $\Phi = (\Phi[1], \Phi[2], \dots \Phi[n'])$ be a sequence of $\eps$-synchronization hash functions with respect to $x$, and $\Pi$ be a matching between $x_T$ and $x$ under $\Phi$.
	We say a pair $(i, j) \in \Pi$ is a good pair, if $x_T[i] = x[j, j+T)$. Otherwise we say it is a bad pair.
	Let $g$ be the number of good pairs in $\Pi$,
	and $b$ be the number of bad pairs in $\Pi$.
	Then
	\begin{equation*}
	b < 2\eps(n' - g).
	\end{equation*}
\end{lemma}

\begin{proof}
	As the string $x$ is \ref{Bdistinct},
	a pair $(i, j)$ is good if and only if $j = T(i - 1) + 1$.
	For any good pair $(i, j)$,
	we remove the block $x_T[i]$ in $x_T$ and $x[j, j+T)$ in $x$.
	Then we list all bad pairs in $\Pi$ $(i_1, j_1), (i_2, j_2), \dots, (i_b, j_b)$ such that $i_1 < i_2 < \dots < i_b$.
	We use the bad pairs to divide the strings into disjoint intervals such that each bad pair is contained in an interval as in Definition~\ref{epssyncfuncdef}, as follows.
	There are two cases: $T(i_1 - 1) + 1 > j_1$ and $T(i_1 - 1) + 1 < j_1$.
	
	\noindent
	\begin{figure}[H]
	\centering
	\includegraphics[width=0.5\textwidth]{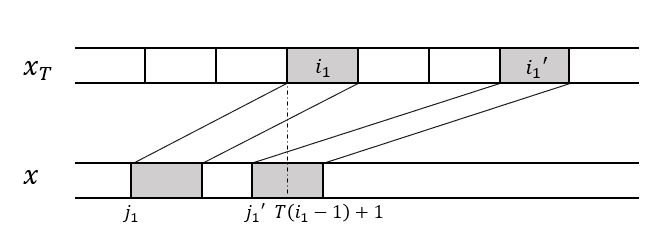}
	\caption{Pictorial representation for the case $T(i_1 - 1) + 1 > j_1$.}
	\label{fig-case1}
	\end{figure}
	If $T(i_1 - 1) + 1 > j_1$, then we find the largest $j_1'$ such that $j_1' < T(i_1 - 1) + 1$ and there exists an $i_1'$ such that $(i_1', j_1') \in \Pi$.
	Let $b_1$ be the number of bad pairs $(i, j)$ in $\Pi$ satisfying $i_1 \le i \le i_1'$, denote $l_1^{(1)} = i_1' - i_1 + 1$ and $l_2^{(1)} = T(i_1-1) - j_1 + 1$
	(See \figurename~\ref{fig-case1}).
	Applying Equation~\ref{eq-sync-1},
	we obtain
	\begin{equation*}
	b_1 \le \mathsf{MATCH}_{\Phi}(x_T[i_1, i_1'], x[j_1, T(i_1 - 1)]) <
	\eps (l_1^{(1)} + l_2^{(1)}/T)
	\end{equation*}
	
	\noindent
	\begin{figure}[H]
	\centering
	\includegraphics[width=0.5\textwidth]{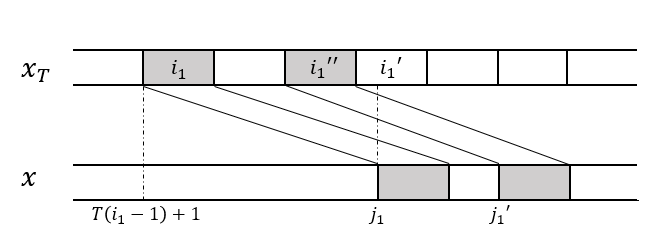}
	\caption{Pictorial representation for the case $T(i_1 - 1) + 1 < j_1$.}
	\label{fig-case2}
	\end{figure}
	If $T(i_1 - 1) + 1 < j_1$, we find the largest $i_1'$ satisfying $T(i_1' - 1) + 1 < j_1$. Then find the largest $i_1''$ such that $i_1'' \le i_1'$ and there exists $j_1'$ satisfying $(i_1'', j_1') \in \Pi$.
	Let $b_1$ be the number of the bad pairs $(i, j)$ in $\Pi$ satisfying $i_1 \le i \le i_1'$, and denote $l_1^{(1)} = i_1' - i_1 + 1$ and $l_2^{(1)} = j_1' - T(i_1' - 1) - 1$ (See \figurename~\ref{fig-case2}).
	Applying Equation~\ref{eq-sync-2},
	we obtain
	\begin{equation*}
	b_1 \le \mathsf{MATCH}_{\Phi}(x_T[i_1, i_1'], x(T(i_1'-1)+1, j_1']) < \eps (l_1^{(1)} + l_2^{(1)}/T)
	\end{equation*}

	Continue doing this for the bad pairs $(i, j)$ satisfying $i > i_1'$, we obtain $b_k < \eps (l_1^{(k)} + l_2^{(k)}/T)$ for some $k = 1, 2, \dots$.
	Adding them up, we have
\begin{equation}
b = \sum_k b_k < \eps\left(\sum_k l_1^{(k)} + \sum_k \frac{l_2^{(k)}}{T}\right)
\end{equation}
On the right hand side,
as the intervals corresponding to $l_1^{(k)}, l_2^{(k)}$ are disjoint and there are $g$ good pairs,
we have $\sum_k l_1^{(k)} \le n' - g$ and $\sum_k l_2^{(k)} / T \le (n - Tg) / T = n' - g$.
Hence we obtain $b < 2\eps (n' - g)$.
\end{proof}

\begin{theorem}\label{epshashtheorem}
	For any $0 < \eps < 1$ and $T \in \N, T \geq B$, let $x$ be a \ref{Bdistinct} string of length $n = Tn'$, and
	$x_T = (x[1, T], x[T+1, 2T], \dots, x[(n'-1)T+1, n])$ be a partition of $x$.
	Let $\Phi = (\Phi[1], \Phi[2], \dots \Phi[n'])$ be a sequence of $\eps$-synchronization hash functions with respect to $x$,
	and $y$ is a string satisfying $ED(x, y) \le k$.
	If $\Pi = \mathsf{MATCH}_{\Phi}(x_T, y)$,
	and denote the number of bad pairs in $\Pi$ as $b$, then
	\begin{equation}
		b < \frac{1 + 2\eps}{1 - 2\eps} k
	\end{equation}
\end{theorem}

\begin{proof}
	As $ED(x, y) \le k$, there is a matching of size $n' - k$ between $x_T$ and $y$ under $\Phi$. Hence, $|\Pi| \ge n' - k$.
	From Theorem~\ref{matchToSelfMatch}, there exists a self matching $\Pi'$ between $x_T$ and $x$ under $\Phi$. Let $b'$ be the number of bad pairs in $\Pi'$ and $g'$ be the number of good pairs in $\Pi'$. Then the theorem guarantees
	$|\Pi'| \ge |\Pi| - k$ and $b' \ge b - k$.
	From Lemma~\ref{epshashlemma}, $b' < 2\eps (n' - g')$. On the other hand, we have
	\begin{equation}
		b' + g' = |\Pi'| \ge |\Pi| - k \ge n' - 2k
	\end{equation}
	So we obtain $b' < 2\eps (n' - g') \le 2\eps(b' + 2k)$. Manipulating it, we derive an uppers bound for $b'$.
	\begin{equation}
		b' < \frac{4\eps}{1 - 2\eps} k
	\end{equation}
	Hence we have
	\begin{equation}
		b \leq b' + k < \frac{1 + 2\eps}{1 - 2\eps} k
	\end{equation}
\end{proof}

\subsection{Construction}

In the first stage, the two parties use a fixed small string $p = 1 \circ 0^{s-1}$ of length $s$,
and find all $p$-split points in their strings.
As the string $x$ is uniform random, with high probability,
the distance between any two adjacent $p$-split points is $O(s 2^s \log n)$.
But some $p$-split points may be too close to each other.
So we only choose the $p$-split points $i$ such that the next $p$-split point of $i$ is at least $2^s/2$ away,
and use these chosen $p$-split points to partition the string into blocks.

\begin{construction}[\textbf{Stage I}]\label{firststage}
	Let $n$ denote the input length.
	Parameters $s = \log \log n + 3, B = 3 \log n, T_0 = s 2^s \log n$, and $p = 1 \circ 0^{s-1}$ is a fixed string of length $s$. To make the representation simple, we assume $n$ is a multiple of $T_0$.
	
	Alice: On input uniform random string $x \in \{0, 1\}^n$.
	\begin{enumerate}[label*=\arabic*.]
		\item Choose all $p$-split points $i$ of $x$ such that its next $p$-split point $j$ satisfies $j - i > 2^s/2$.
		Denote the chosen $p$-split points as $i_1, i_2, \dots, i_{n'}$.
		Partition the string $x$ into blocks $x[1, i_2)$, $x[i_2, i_3)$, $x[i_3, i_4)$ $\dots, x[i_{n'}, n]$,
		and index these blocks as $1, 2, 3, \dots, n'$.
		
		\item Create a set $V = \left\{(\textsf{len}_b, \textsf{B-prefix}_b, \textsf{B-prefix}_{b+1}) \mid 1 \le b \le n' - 1\right\}$,
		where $\textsf{len}_b$ is the length of the $b$-th block, and $\textsf{B-prefix}_b$ and
		$\textsf{B-prefix}_{b+1}$ are the B-prefix of the $b$-th block and
		the $(b+1)$-th block respectively.
		
		\item Represent the set $V$ as its indicator vector, which has size $\poly(n)$, 
		and send the redundancy $z_V$ being able to correct $4 k$ Hamming errors, using Theorem~\ref{agcode} (or simply using a Reed-Solomon code).
		
		\item Partition the string $x$ evenly into $n / T_0$ blocks, each of size $T_0$. \label{AlicelastI}
	\end{enumerate}
	
	Bob: On input string $y \in \{0, 1\}^{m}$ satisfying $ED(x, y) \le k$, and the redundancy $z_V$ sent by Alice.
	\begin{enumerate}[label*=\arabic*.]
		\item Choose the $p$-split points $i$ of $y$ such that its next $p$-split point $j$ satisfies $j - i > 2^s/2$.
		Denote the chosen $p$-split points as $i_1', i_2', \dots, i_{m'}'$.
		Partition the string $y$ into blocks $y[1, i_2')$, $y[i_2', i_3')$, $y[i_3', i_4')$ $\dots, y[i_{m'}', n]$,
		and index these blocks as $1, 2, 3, \dots, m'$.
		
		\item Create a set $V' = \left\{(\textsf{len}_b, \textsf{B-prefix}_b, \textsf{B-prefix}_{b+1}) \mid 1 \le b \le m' - 1\right\}$
		using the partition of $y$.
		
		\item Use the indicator vector of $V'$ and the redundancy $z_V$ to recover Alice's set $V$.
		
		\item Create an empty string $\tilde{x}$ of length $n$,
		and partition $\tilde{x}$ according to the set $V$ in the following way:
		first find the element $(\textsf{len}^{(1)}, \textsf{B-prefix}^{(1)}, \textsf{B-prefix}'^{(1)})$
		in $V$ such that
		for all elements $(\textsf{len}, \textsf{B-prefix}, \textsf{B-prefix}')$ in $V$,
		$\textsf{B-prefix}^{(1)} \neq \textsf{B-prefix}'$.
		Then partition $\tilde{x}[1, \textsf{len}^{(1)}]$ as the first block,
		and fill $\tilde{x}[1, B]$ with $\textsf{B-prefix}^{(1)}$. Then find the element $(\textsf{len}^{(2)}, \textsf{B-prefix}^{(2)}, \textsf{B-prefix}'^{(2)})$
		such that $\textsf{B-prefix}^{(2)} = \textsf{B-prefix}'^{(1)}$,
		and partition $\tilde{x}[\textsf{len}^{(1)}+1, \textsf{len}^{(1)} + \textsf{len}^{(2)}]$ as the second block, and fill $\tilde{x}[\textsf{len}^{(1)}+1, \textsf{len}^{(1)} + B]$ with $\textsf{B-prefix}^{(2)}$.
		Continue doing this until all elements in $V$ are used to recover the partition of $x$.
		
		\item \label{Bobfillblocks} For each block $b$ in $\tilde{x}$,
		if Bob finds a unique block $b'$ in $y$ such that the B-prefix of $b'$ matches the B-prefix of $b$ and the lengths of $b$ and $b'$ are equal,
		Bob fills the block $b$ using $b'$.
		If such $b'$ doesn't exist or Bob has multiple choices of $b'$,
		then Bob just leaves the block $b$ as blank. 
		
		\item Partition the string $\tilde{x}$ evenly into $n / T_0$ blocks, each of size $T_0$. \label{BoblastI}
	\end{enumerate}
	
\end{construction}

\begin{construction}[\textbf{Stage II}]
	The second stage consists of $L = \left\lceil \frac{\log( O(s 2^s) )}{\log (\log^{0.4} n)} \right\rceil = O(1)$ levels.
	Let $T' = \log^{0.6} n, T'' = \log^{0.4} n$, $T_{l} = T_{l-1} / T''$
	for $1 \le l \le L-1$.
	In the last level, we choose $T_L = B$ and $T_{L} \ge T_{L-1} / \log^{0.4} n$.
	To make the representation simple,
	in this stage we assume $n$ is a multiple of $T_l$, for all $l \in [L]$. \\

	Alice: For $l = 1, 2, \dots L$,
	in $l$-th level,
	\begin{enumerate}
		\item Partition the string $x$ evenly into $n'_l = n / T_l$ blocks, each of size $T_l$.
		\item Applying construction~\ref{epssynchashfunctionconstruction} with block size $T=T_l$,
		Alice gets a sequence of $\eps$-synchronization hash functions $\Phi = (\Phi[1], \Phi[2], \dots, \Phi[n'_l])$ with respect to $x$, where each $\Phi[t], t \in [n'_l]$ consists of a pair of functions $(\phi[t], \theta[t])$.
		\item Alice sends the description of $(\phi[t])_{t \in [n'_l]}$ to Bob. By Lemma~\ref{philemma}, the description uses $O(\log n)$ bits.
		
		\item Alice packs every successive $T'$ elements of $(\theta[t])_{t \in [n'_l]}$ into a vector $V_{\theta}$, i.e. $V_{\theta} = (\theta[1,T'], \theta[T'+1, 2T'], \dots)$. Note that by Lemma~\ref{thetalemma}, each $\theta[t]$ has a description of size $O(\log \log n)$.
		Alice sends the redundancy $z_{\theta}$ being able to correct some $O(k)$ Hamming errors of $V_{\theta}$, using Theorem~\ref{agcode}.
		
		\item For any $t \in [n'_l]$,
		Alice evaluates $\Phi[t]$ on the $t$-th block of $x$, and obtains the hash values $I[t] = \Phi[t](x[T_l(t-1)+1, T_l(t-1) + B])$, and stores $(I[t])_{t \in [n'_l]}$ into a vector $I$.
		Then she packs every successive $T''$ elements of $I$ into a vector $V_I$, i.e. $V_{I} = (I[1, T''], I[T''+1, 2T''], \dots)$,
		and sends the redundancy $z_I$ being able to correct some $O(k)$ Hamming errors of $V_I$, using Theorem~\ref{agcode}.
	\end{enumerate}
	After $L$ levels,
	Alice evenly partitions her string $x$ into $n / T_L$ small blocks, each of size $T_L$. Alice then sends a redundancy $z_x$ being able to correct some $O(k)$ wrong blocks or unmatched blocks, using Theorem~\ref{agcode}. \\
	
	Bob: For $l = 1, 2, \dots, L$,
	in the $l$-th level, Bob receives the description of the functions $(\phi[t])_{t \in [n'_l]}$, and the redundancies $z_{\theta}, z_I$. Finally he receives $z_x$.
	\begin{enumerate}
		\item Partition the string $\tilde{x}$ evenly into $n'_l = n / T_l$ blocks, each of size $T_l$.
		
		\item
		For any $t \in [n'_l]$, denote $S_t = \left\{\tilde{x}[u, u+B) \mid |T(t-1)+1 - u| < T_l \log^{0.6} n, 1 \le u \le n - B + 1\right\}$.
		Bob applies Lemma~\ref{thetalemma} using $S_t$ for any $t \in [n'_l]$, and obtains $(\theta'[t])_{t\in[n'_l]}$.

		\item Bob packs every successive $T'$ elements of $(\theta'[t])_{t \in [n'_l]}$ into a vector $V_{\theta}' = (\theta'[1, T'], \theta'[T'+1, 2T'], \dots)$. Then he uses the redundancy $z_{\theta}$ and $V_{\theta}'$ to recover $V_{\theta}$. Bob unpacks $(\theta[t])_{t \in [n'_l]}$ from $V_{\theta}$ to
		obtain $\Phi$.
		
		\item
		For any $t \in [n'_l]$, Bob evaluates the hash function $\Phi[t]$ on the $t$-th block of $\tilde{x}$, so he obtains the hash values $I'[t] = \Phi[t](\tilde{x}[T_l(t-1)+1, T_l(t-1)+B])$
		. Bob packs every successive $T''$ elements of $I'$ into the a vector $V_{I}'$, i.e. $V_{I}' = (I'[1, T''], I'[T''+1, 2T''], \dots)$.
		
		\item Bob uses the redundancy $z_I$ and $V_{I}'$ to recover $V_I$, then obtains $I$ from $V_I$ by unpacking.
		
		\item Bob finds the maximum matching $\Pi$ between $x$ and $y$ under $\Phi$ using $I$. Note that in order to find such a matching, Bob only needs to know the hash values of $\Phi[t]$ on the $t$-th block of $x$, which can be obtained from the vector $I$.
		For each pair $(a, b)$ in $\Pi$, Bob fills the $a$-th block $\tilde{x}[T_l(a-1)+1, T_l a]$ with $y[b, b + T_l)$.
	\end{enumerate}
After $L$ levels,
Bob partitions $\tilde{x}$ evenly into blocks of length $T_L$, then uses the redundancy $z_x$ to recover $x$.

\end{construction}

\subsection{Analysis}

\begin{lemma} \label{blocklengthlemma}
	For Stage I,
	in Alice's partition of $x$ and
	Bob's partition of $y$, the size of every block is greater than $2^s / 2$.
	If all properties in Theorem~\ref{Bdistinctandinterval} hold,
	then in Alice's partition of $x$,
	the size of every block is at most $B_1 + B_2 = O(s 2^s \log n)$.
\end{lemma}

\begin{proof}
	We first prove that the size of every Alice's block is greater than $2^s / 2$. The sizes of Bob's blocks follow the same routine.
	For any $t = 1, 2, \dots n'-1$, let $j$ be the next $p$-split point of $i_t$, then $j - i_t > 2^s / 2$.
	As $i_{t+1}$ is also a $p$-split point, from the Definition ~\ref{nextsplitpoint}, we have $i_{t+1} \ge j$. Hence, $i_{t+1} - i_t \ge j - i_t > 2^s / 2$. For the first block $x[1, i_2]$, its length is $i_2$, which is greater than $i_2 - i_1 > 2^s / 2$. For the last block $x[i_{n'}, n]$, let $j$ be the next $p$-split point of $i_{n'}$, then $j \le n+1$ and $j - i_{n'} > 2^s / 2$.
	Hence the length of this block is $n - i_{n'} + 1$, which is greater than $2^s / 2$.
	
	For sake of contradiction, assume there is a block $x[l, r)$ whose length is greater than $B_1 + B_2$, then by \ref{uniformproperty1}, there is a $p$-split point $j$ in the range of $[l + 1, l + B_1]$. Applying \ref{uniformproperty2},
	there is a $p$-split point chosen by Alice in the range of $[j, j+B_2)$. This violates the assumption that $x[l, r)$ is a block.
\end{proof}

\begin{lemma}\label{symdifflemma}
	If $2^s / 2 > B$, then the symmetric difference of $V$ and $V'$ has size at most $4k$.
\end{lemma}

\begin{proof}
	W.L.O.G, we only need to prove that
	the size of $V \setminus V'$ is at most $2k$.
	
	As $ED(x, y) \le k$, there exists a
	series of $k$ edit operations transforming $x$ into $y$.
	For any element $(\mathsf{len}_b, \textsf{B-prefix}_b, \textsf{B-prefix}_{b+1}) \in V$, where $b$ is a block index in the partition of $x$,
	if the $b$-th block and the $(b+1)$-th block are not involved in the edit operations,
	then the $b$-th block is still a block  in Bob's string $y$. By Lemma~\ref{blocklengthlemma},
	the sizes of $b$-th block and $(b+1)$-th block is at least $2^s / 2 = 4 \log n > B$, and so
	$\mathsf{len}_b$, $\textsf{B-prefix}_b$ and $\textsf{B-prefix}_{b+1}$ remain the same in Bob's string $y$. So we have $(\mathsf{len}_b, \textsf{B-prefix}_b, \textsf{B-prefix}_{b+1}) \in V'$. Hence, an element
	$(\mathsf{len}_b, \textsf{B-prefix}_b, \textsf{B-prefix}_{b+1}) \in V \setminus V'$
	implies that the $b$-th block or the $(b+1)$-th block is involved in the edit operations.
	So the size of $V \setminus V'$ is upper bounded by $2k$.
\end{proof}

\begin{theorem} \label{stageIthm}
	If $2^s/2 > B$, and all properties in Theorem~\ref{Bdistinctandinterval} hold,
	then after Stage I,
	at most $O(k)$ blocks of $\tilde{x}$ contains unfilled bits or incorrectly filled bits.
\end{theorem}

\begin{proof}
	By Lemma~\ref{symdifflemma},
	Bob recovers $V$ correctly using the redundancy $z_V$.
	
	As $ED(x, y) \le k$, there exists a
	series of $k$ edit operations transforming $x$ into $y$.
	In Bob's step \ref{Bobfillblocks},
	for every block $b$,
	if Bob does not find a block $b'$ in $y$ that matches the $B$-prefix and the length of $b$, then the block $b$ must be involved in an edit operation. Hence the number of such blocks is at most $k$.
	
	For the case where Bob finds at least two blocks
	$b_1'^{(b)}$ and $b_2'^{(b)}$ in $y$ such that
	the B-prefix and the length of $b_1'^{(b)}$, $b_2'^{(b)}$ both match that of $b$,
	one of $b_1'^{(b)}$ and $b_2'^{(b)}$ must be involved in an edit operation, otherwise they are both substrings of $x$ and thus violate the \ref{Bdistinct} property.
	By the \ref{Bdistinct} property of the string $x$,
	$\left\{b_1'^{(b)}, b_2'^{(b)}\right\}_b$ are disjoint sets for different block $b$.
	As the total number of edit operations is at most $k$, the number of the blocks $b$ that Bob finds multiple $b'$ is at most $k$.
	
	For the blocks Bob fills in,
	at most $k$ of them are involved in edit operations, these blocks may be incorrectly filled. But for the remaining blocks, they must be correctly filled due to the \ref{Bdistinct} property.
	In total, the number of unfilled or incorrectly filled blocks is at most $3k$.
	
	In Bob's step \ref{BoblastI},
	Bob divides $\tilde{x}$ evenly into $n / T_0$ blocks, each of length $T_0$. 
	By Lemma~\ref{blocklengthlemma},
	the length of each block in step~\ref{Bobfillblocks} is $O(s 2^s \log n)$. As we set $T_0 = s 2^s \log n$, each unfilled or incorrectly filled block in step \ref{Bobfillblocks} can only affect $O(s 2^s \log n) / T_0 = O(1)$ blocks in step \ref{BoblastI}.
	Hence, the total number of blocks in step~\ref{BoblastI} containing unfilled or incorrectly filled bits is at most $O(k)$.
\end{proof}


\begin{theorem} \label{stageIItheorem}
	If $2^s/2 > B$, $\eps \leq 1/3$ and all properties in Theorem~\ref{Bdistinctandinterval} hold,
	then at the end of each level in stage II, the total number of unmatched blocks and incorrectly matched blocks between $x$ and $\tilde{x}$ is $O(k)$, where the constant hidden in big $O$ notation is independent of the levels.
\end{theorem}

\begin{proof}
	We will prove by induction.
	We denote the the partition of $\tilde{x}$ in $0$-th level to be the partition in the last step of Stage I.
	For the $l$-th level, $l \ge 1$,
	we define the \emph{bad block} to be the block of $\tilde{x}$ containing unfilled or incorrectly filled bits in the $(l-1)$-th level. 
	For the $0$-th level, by Theorem~\ref{stageIthm}, the number of the bad blocks is at most $O(k)$, and thus the theorem holds for $0$-th level.
	Now we assume the theorem holds for the $(l-1)$-th level, then the number of the bad blocks in the $(l-1)$-th level is bounded by $O(k)$.
	As we set $T_l = T_{l-1} / T''$,
	the length of each bad block is at most $O(T_{l} \cdot T'')$.
	Each of the bad block is divided evenly into $O(T'')$ smaller successive blocks in step $1$ of the $l$-th level.
	
	For each $(\theta[t])_{t \in [n'_l]}$, $\theta[t]$ is computed deterministically by $S_t$ in Construction~\ref{epssynchashfunctionconstruction}, where $S_t$ contains all substrings of length $B$ of $x$ in the range of $[T_l(t-1)+1 - T_l\cdot \log^{0.6} n, T_l(t-1)+1 + T_l\cdot \log^{0.6} n]\cap[1, n]$.
	As we pack every $T' = \log^{0.6} n$ successive $\theta[t]$ into one element in $V_{\theta}$, and one block in $x$ contains $T_l$ bits, one bad block can only cause $O(T_l T'' + T_l \log^{0.6}n) / T_l T' = O(1)$ Hamming errors between $V_{\theta}$ and $V_{\theta}'$.
	Hence, there are at most $O(k)$ Hamming errors between $V_{\theta}$ and $V_{\theta}'$. Bob can thus recover Alice's $V_{\theta}$ correctly using the redundancy $z_{\theta}$.
	
	After Bob correctly recovers $\Phi$, he evaluates the $\eps$-synchronization hash function on each block.
	One \emph{bad block} can only cause $O(T'')$ successive Hamming errors between $I$ and $I'$. So after packing every $T''$ successive hash values, the number of the Hamming errors between the vector $V_I$ and $V_{I'}$ is upper bounded by $O(k)$. Hence the redundancy $z_I$ allows Bob to recover $V_I$.
	
	Now Bob obtains the $\eps$-synchronization hash functions $\Phi$ and its hash values,
	he computes $\Pi = \mathsf{MATCH}_{\Phi}(x_T, y)$ using the hash values in $I$. Since $\eps \leq 1/3$, by Theorem~\ref{epshashtheorem},
	the number of the bad pairs is upper bounded by $O(k)$.
	
	Note that we only have constant levels, so the number of unmatched and incorrectly matched blocks in the $L$-th level is still bounded by $O(k)$.
\end{proof}

\begin{proof}[Proof of Theorem~\ref{randdocexmaintheorem}]
	We choose $s = \log \log n + 3$,
	so that $2^s/2 = 4 \log n > B$.\ With probability $1 - 1/\poly(n)$,
	all properties in Theorem~\ref{Bdistinctandinterval} hold. In each level in Stage II, we use a sequence of $\eps$-synchronization hash functions with $\eps=1/3$, by construction~\ref{epssynchashfunctionconstruction}. 
	
	By Theorem~\ref{stageIItheorem},
	after Stage II ends,
	there are $O(k)$ blocks of length $B$ that are incorrectly matched or unmatched between $x$ and $\tilde{x}$.
	Hence, the redundancy $z_x$ allows Bob to recover $x$ correctly.
	
	In Stage I, Alice sends the redundancy $z_V$ to Bob. The dimension of the indicator vector of the set $V$ is $2^{2B} \poly \log n = \poly (n)$. If we use Reed-Solomon code or Theorem~\ref{agcode} to generate $z_V$, the size of the redundancy $z_V$ is $O(k \log n)$. 
	
	Stage II has a constant number of levels. In each level, Alice sends the description of $(\phi[t])_{t \in [n'_l]}, z_{\theta}, z_{I}$ to Bob. Finally Alice sends $z_x$ to Bob.
	By Lemma~\ref{philemma}, the description of $(\phi[t])_{t \in [n'_l]}$ has size $O(\log n)$.
	For the redundancy $z_{\theta}$,
	the size of each element in $V_{\theta}$ is $O(\log^{0.6} n) \cdot O(\log \log n)$ bits, which is smaller than $O(\log n)$ bits. Hence, the size of $z_{\theta}$ is $O(k \log n)$.
	For $z_I$, the hash value of $\phi[t], (t \in [n'_l])$ can be stored in $R$ bits, and the hash value of $\theta[t], (t \in [n'_l])$ can be stored in $O(\log \log n)$ bits, so the size of each element in $V_I$ is $O(\log^{0.4} n) \cdot (R + O(\log \log n)) < O(\log n)$ bits. Hence, the size of $z_I$ is $O(k \log n)$ bits.
	For $z_x$, its size is $O(k B) = O(k \log n)$ bits. Thus in total the size of the redundancy is $O(k \log n)$.
\end{proof}

\section{Explicit binary ECC for edit errors}
\label{sec:InsdelCode}
In this section we'll show how to use the document exchange protocol for uniform random strings in Section \ref{sec:randdocexc} to construct ECC for edit errors.

Our general strategy is as follows. For any given message $x \in \{0,1\}^n$, we use a generator to generate a mask string s.t. the xor of the mask and the message, say $y(U_r) \in \{0,1\}^n$, has the three properties: \ref{uniformproperty1}, \ref{uniformproperty2}, \ref{Bdistinct}. We ensure that the seed length for the generator is small enough s.t. we can exhaustively search the seed $u \in \{0,1\}^r$ s.t. $y(u)$ has the three properties. Then we apply the method in Section \ref{sec:randdocexc} to create a redundancy $z$ for $y(u)$.  After that we use an asymptotically good binary ECC for edit errors to encode the redundancy and the seed. Concatenating this with $y$ gives the final codeword.

\subsection{The generator for the mask}
We show that there exists an explicit generator of seed length $O(\log n)$ s.t. given any message $x \in \{0,1\}^n$, w.h.p. the xor of $x$ and the output of the generator has \ref{uniformproperty1}, \ref{uniformproperty2} and \ref{Bdistinct}, where the randomness is over the uniform random seed of the generator. Formally we have the following theorem.

\begin{theorem}
\label{gfor3property}
There exists an algorithm (generator) $g$ s.t. for every $  n \in \mathbb{N},  x\in \{0,1\}^n$,  with probability $1-1/\poly(n)$,  $g(n, U_r) + x $ satisfies \ref{uniformproperty1}, \ref{uniformproperty2} and \ref{Bdistinct},    where  $r= O(\log n)$. (Let $s$ in \ref{uniformproperty1}, \ref{uniformproperty2} be $ \log \log 
n+O(1)$.)
\end{theorem}

$g$ is the xor of three generators, each of which generates a string satisfying one of the three properties. We utilize random walks on expander graphs and PRG for $\AC^0$ circuits to reduce the seed length of generators for \ref{uniformproperty1} and \ref{uniformproperty2}. And we use almost $\kappa$-wise independence generator for \ref{Bdistinct}.








\begin{lemma}
\label{g1}
There exists an algorithm (generator) $g_1$ s.t. for every $   n \in \mathbb{N},  x\in \{0,1\}^n$,  with probability $1-1/\poly(n)$,  $g_1(n, U_{r_1}) + x  $ satisfies \ref{uniformproperty1}, where  $r_1 = O(\log n)$.
\end{lemma}

\begin{proof}

Regard $x$ as a sequence of  blocks, each of length $ B_1 = O(s 2^s \log n) $. W.l.o.g. we assume every $n$ is dividable by $B_1$. Since if not, we can simply ignore the last block whose length is less than $B_1$. Denote $t = n/B_1$.

We show how to generate a block $z \in \{0,1\}^{B_1}$ s.t.   $ x[1] + z$ contains a $p$-split point with probability $1-1/\poly(n)$.

We further divide $x[1]$ into a sequence of $t_0 = B_1/b_0 = O(\log n)$ smaller blocks, each of length $b_0 = \Theta(s 2^s)$.
Correspondingly we view $z$ as a sequence of $t_0$ blocks, each of length $b_0$.

For every $i\in [t_0]$, if $z[i]$ is uniform, then the probability that there is a $p$-split point in $ x[1][i] + z[1]$ is at least $2/3$. This is because for every consecutive $s$ bits, the probability that it is $p$ is $\frac{1}{2^s}$. Since there are $b_0/s$ number of distinct $s$-bits substrings in $ x[1][i] + z[1]$, the probability that none of them is $p$ is $(1-\frac{1}{2^s})^{b_0/s}$. So the probability that at least one of them is $p$ is  $1 - (1-\frac{1}{2^s})^{b_0/s} = 1 - (1-\frac{1}{2^s})^{O(2^s)}\geq 2/3$ if we let the constant in $b_0$ to be large enough. Note that checking whether $ x[1][i] + z[1]$ has a $p$-split point can be done by a CNF/DNF $f$ of size $m = O(b_0 s)$ since we can set up a test for every consecutive $s$-bits substring and each test checks whether the corresponding $s$-bits substring is $p$.  So if we instead use the generator $\tilde{g}$ from Theorem
\ref{PRGforCNF}, with error $\tilde{\eps} = 1/6 $, to generate $z[i]$, i.e. $z[i] = \tilde{g}(b_0, m = O(b_0), 1/6 , U_{\tilde{r}})$ then
\[\Pr[f(x[1][i]+z[i]) = 1]    \geq   \Pr[f(x[1][i]+U_{b_0}) = 1] - \tilde{\eps} \geq 2/3 - 1/6 = 1/2, \]
where  $\tilde{r} = O(\log^{2} m  \log \log m ) = \poly(\log \log n)$.

We generate $z$ by doing a random walk of length $ t_0 = O(\log n) $, on an $(2^{\tilde{r}}, \Theta(1), \lambda)$-expander graph with constant $\lambda \in (0,1)$.  The expander we use is from the strongly explicit expander family of Theorem \ref{explicitexpander}, so the random walk can be done in polynomial time.

Denote the vertices reached sequentially in the random walk as $w_1, w_2, \ldots, w_{t_0} \in \{0,1\}^{\tilde{r}}$.
The total length of random bits used here is $r_1 = \tilde{r} + O(t_0) = O(\log n)$. Note that for every $i$, the probability that $x[1][i] + z[i]$ contains a $p$-split point is  $\Pr[f(x[1][i] + z[i]) = 1] \geq 1/2$. By Theorem \ref{ExapanderRWHitting}, the probability that  $ x[1] + z$ does not contain a $p$-split point is
\[  \Pr[\forall i,  f(x[1][i] + z[i]) = 0  ] \leq  (\frac{1}{2}+\lambda)^{O(t_0)} = 1/\poly(n),
\]
if we pick the constant in $t_0$ to be large enough and $\lambda$ to be some small constant e.g. $1/3$.

Let $g_1(n, U_{r_1}) $ be the concatenation of $t$ number of strings $z$.\ By a union bound, with probability $1-t/\poly(n) = 1 - 1/\poly(n)$, every block of $x + g_1(n, U_{r_1})$ has a $p$-split point. Note that $r_1 = \tilde{r}+ O(\log n) = O(\log n)$.
\end{proof}

Next we give a generator for \ref{uniformproperty2}.
\begin{lemma}
\label{g2}
There exists an algorithm (generator) $g_2$ s.t. for every $ n \in \mathbb{N}, \forall x\in \{0,1\}^n$,  with probability $1-1/\poly(n)$,  $g_1(n, U_{r_2}) + x \in \{0, 1\}^{n}$ satisfies \ref{uniformproperty2},  where  $r_2 = O(\log n)$.
\end{lemma}

\begin{proof}

Again let's view $x$ as a sequence of  blocks, each of length $ B_2 = O( 2^s \log n) $. W.l.o.g. we assume every $n$ is dividable by $B_2$, since if not, we can simply ignore the last block which is not a whole block. Denote $t = n/B_2$.

We show how to generate a block $z \in \{0,1\}^{B_2}$ s.t.   with probability $1-1/\poly(n)$, $ x[1] + z$ contains an interval of length at least $ 2^s/2 + s $ which does not contain a substring $p$.

We further divide $x[1]$ into a sequence of $t_0 = B_2/b_0 = O(\log n)$ smaller blocks, each of length $b_0 =  2^s/2 +s $.
Correspondingly we view $z$ as a sequence of $t_0$ blocks, each of length $b_0$.

For every $i\in [t_0]$, if $z[i]$ is uniform, then we claim that the probability that
$x[1][i]+z[i]$ contains a substring $p$ is at most $\frac{1}{2^s} \cdot ( \frac{2^s}{2}  ) = \frac{1}{2}  $. This is because for a fixed consecutive $s$ bits the probability, that it is $p$, is $\frac{1}{2^s}$. There are at most $ \frac{2^s}{2}  $ different intervals of length $s$. The claim holds by a union bound.

We generate $z$ by doing a random walk, of length $ t_0 = O(\log n) $, on an $(2^{b_0}, \Theta(1), \lambda)$-expander graph with constant $\lambda \in (0,1)$.  The expander we use is from the strongly explicit expander family of Theorem \ref{explicitexpander}, so the random walk can be done in polynomial time.

Denote the vertices reached sequentially in the random walk as $w_1, w_2, \ldots, w_{t_0} \in \{0,1\}^{b_0}$.
The total length of random bits used here is $r_2 = b_0 + O(t_0) = O(\log n)$. Note that for every $i$, the probability that $x[1][i] + z[i]$ contains  $p$ is  at most $1/2$. By Theorem \ref{ExapanderRWHitting}, the probability that  every block of  $ x[1] + z$   contains   $p$  is
\[  \Pr[\forall i,  x[1][i] + z[i] \mbox{ contains a substring } p  ] \leq  (\frac{1}{2}+\lambda)^{O(t_0)} = 1/\poly(n),
\]
if we pick the constant in $t_0$ to be large enough and $\lambda$ to be some small constant e.g. $1/3$.

Let $g_2(n, U_{r_2}) = z^t$.
By a union bound, the probability that every block of $x + g_2(n, U_{r_2})$ has at least $1$ sub-block which does not contain $p$ as a substring, is at least $1-t/\poly(n) = 1 - 1/\poly(n)$.

Note that for every interval of length $2B_2+s$ starting at a $p$-split point $i$, the last $2B_2$ bits must contain a block of length $b_0$ which does not contain substring $p$. Assume it's $x[j, j']$. We can find the left closest $p$-split point to the left end of $x[j, j']$. Assume its index is $j_1$. Note that this means $ i \leq j_1 \leq j-1 $. Assume the right closest $p$-split point to $x[j'-s+1] $ is $ j_2 $. We know $j_2 > j'-s+1 $. So $j_2 - j_1 > j'-s+1 -j+1 = b_0 + 1 - s > 2^s/2$.

Thus \ref{uniformproperty2} is satisfied.
\end{proof}

\begin{lemma}
\label{g3}
There exists an algorithm (generator) $g_3$ s.t. for every $  n \in \mathbb{N},  x\in \{0,1\}^n$,  with probability $1-1/\poly(n)$,  $g_3(n, U_r) + x  $ is \ref{Bdistinct},  where  $r_3 = O(\log n)$.
\end{lemma}

\begin{proof}

Let $g_3$ be the $\eps$-almost $\kappa$-wise independence generator from Theorem \ref{almostkwiseg}, where $ \eps = 1/\poly(n)$, $\kappa = 2B$, seed length $r_3 = O(\log \frac{\kappa \log n}{ \eps })= O(\log n)$.

Given a fixed $x$, $g_3(n, U_{r_3})+x$ is $\eps$-almost $\kappa$-wise independent. So for every pair of two intervals $u, v \in \{0,1\}^{B}$ of it,
\[
	\Pr[u  = v] = \sum_{a \in \{0,1\}^{B}} \Pr[ u = a, v = a ] \leq \sum_{a \in \{0,1\}^{B}} (\frac{1}{2^{2B}} + \eps) \leq \frac{1}{2^B} + 1/\poly(n) \leq \frac{1}{2^{B-1}},
\]
where the first inequality is due to the definition of $\eps$-almost $\kappa$-wise independence. The second inequality holds since $\eps = 1/\poly(n)$ is small enough. The third inequality is because  $\eps$ is small enough and we can take the constant in $ B =O(\log n)$ to be large enough.
By a union bound over all $O(n^2)$ pair of $u,v$, it concludes that $g_3(n, U_{r_3}) + x \in \{0, 1\}^{n}$ is \ref{Bdistinct} with probability $1- 1/\poly(n)$ since $B$ is large enough.

\end{proof}

Now we can prove Theorem \ref{gfor3property}.
\begin{proof}[Proof of Theorem \ref{gfor3property}]

Let $g(n, U_r) = g_1(n, U_{r_1}) + g_2(n, U_{r_2}) + g_3(n, U_{r_3})$, where $ U_r = U_{r_1} \circ U_{r_2}\circ U_{r_3}$ is uniform random, $r = r_1 + r_2 + r_3 = O(\log n)$.

Fix $x \in \{0,1\}^n$. Note that $U_{r_1}, U_{r_2}, U_{r_3}$ are independent. So we have the following.

By Lemma \ref{g1}, with probability $1-1/\poly(n)$,  $g(n, U_r) + x  = g_1(n, U_{r_1})+ ( x + g_2(n, U_{r_2}) + g_3(n, U_{r_3}) )$ has \ref{uniformproperty1}.

By Lemma \ref{g2}, with probability $1-1/\poly(n)$, $g(n, U_r) + x  = g_2(n, U_{r_2})+ ( x + g_1(n, U_{r_1}) + g_3(n, U_{r_3}) )$ has \ref{uniformproperty2}.

By Lemma \ref{g3}, with probability $1-1/\poly(n)$, $g(n, U_r) + x  = g_3(n, U_{r_3})+ ( x + g_1(n, U_{r_1}) + g_2(n, U_{r_2}) )$ has \ref{Bdistinct}.

So by the union bound, with probability $1-1/\poly(n)$,  $g(n, U_r) + x \in \{0, 1\}^{n}$ has \ref{uniformproperty1}, \ref{uniformproperty2} and \ref{Bdistinct}.

\end{proof}

\subsection{Constructing binary ECC using redundancies}
Based on our results of document exchange protocols, we can construct binary codes which can correct up to $k$ edit errors. The idea is that the sketch sent by Alice in the document exchange protocol can be used to reconstruct the original message.





\begin{lemma}
\label{concatenationlem}
For every $n,r, k \in \mathbb{N}$, every $x\in \{0,1\}^n, z\in \{0,1\}^r$, if $C$ is a binary code with message length $r$, codeword length $n_C$, that can correct up to $4k$ edit errors, then for every $y\in \{0,1\}^*$ s.t. $\ED(y, x\circ C(z)) \leq k$, one can get $z$ from $y$ using the decoding algorithm of $C$.
\end{lemma}

\begin{proof}

We can run the decoding of $C$ on $y[n+1-k, |y|]$, where $|y| \leq n+n_C + k$. Since there are at most $k$ edit operations,  $\LCS(y[n+1-k, |y|], C(z)) \geq n_C-k$. So
\begin{equation}
\begin{split}
\ED(y[n+1-k, |y|], C(z) ) & = |y[n+1-k, |y|]| + n_C - 2\LCS(y[n+1-k, |y|], C(z)) \\
 & \leq |y[n+1-k, |y|]| + n_C - 2(n_C-k)\\
 & = |y|-(n + 1-k) + 1 + n_C - 2(n_C-k) \\
 & \leq (n+ n_C +k) - (n+1-k)+1 + n_C - 2(n_C-k)\\
 &\leq 4k.
\end{split}
\end{equation}
Thus the decoding can output the correct $z$.

\end{proof}

\begin{theorem}
\label{codeFromSk}

If there exists an explicit document exchange protocol with communication complexity $r(n,k)$, where $n\in \mathbb{N}$ is the input size and $k \in \mathbb{N}$ is the upper bound on the edit distance, then there exists an explicit family of binary ECCs with codeword length $n_C=n+O(r(n,k))$, message length $n$, that can correct up to $k$ edit errors.

\end{theorem}

\begin{proof}
The encoding algorithm is as follows: for message $x\in\{0,1\}^n$, we first compute the redundancy $z$ for $x$ and $3k$ edit errors using the document exchange protocol. Then we encode the redundancy $z$ to be $\tilde{c}$ using the binary ECC from Theorem \ref{asympGoodECCforInsdel} which can correct $\alpha$ fraction of errors with constant rate. Here $\alpha$ is a constant such that $\alpha |\tilde{c}| \geq 4k$. Assume the length of the code is $n_0 = O(r)$.
The final codeword $c$ is the concatenation of the original message $x$ and the encoded redundancy $\tilde{c}$. That is, $c=x\circ \tilde{c}$ and $|c| =  n+ n_0$.

By Lemma \ref{concatenationlem}, we can get the redundancy $z$ from the corrupted codeword $c'$. Note that there are at most $k$ edit errors, $
\ED(c'[1, n+k], x) = n+k + n - 2\LCS(c'[1, n+k], x) \leq n+k+n-2(n-k) \leq 3k$.  Here $\LCS(c'[1, n+k], x) \geq n-k$ is because $c'[1, n+k]$ contains a subsequence of $x$ which are the symbols that are not deleted. There are at most $k$ deletions so the subsequence has length at least $n-k$.

Finally we run the document exchange protocol to compute the original message $x$, where Bob's string is $c'[1, n+k]$ and the redundancy from Alice is $z$.


\end{proof}

\subsection{ Binary ECC for edit errors with almost optimal parameters}

A direct corollary of Theorem \ref{codeFromSk} is the following binary ECC.

\begin{theorem}
For any $n, k \in \mathsf{N}$ with $k \leq n/4$, there exists an explicit binary error correcting code with message length $n$, codeword length $n+O(k \log^2 \frac{n}{k})$ that can correct up to $k$ edit errors. 

\end{theorem}

\begin{proof}

It follows directly from Theorem \ref{codeFromSk} and \ref{deterdocexc}.

\end{proof}

\begin{corollary}

There exists a constant $0<\alpha<1$ such that for any  $0<\eps \leq \alpha$ there exists an explicit family of binary error correcting codes with codeword length $n$ and message length $m$, that can correct up to $k=\eps n$ edit errors with rate $m/n=1- O(\eps \log^2 \frac{1}{\eps})$. 
\end{corollary}

By utilizing the document exchange protocol in Section \ref{sec:randdocexc}, the generator from Theorem \ref{gfor3property} and the concatenation technique of Lemma \ref{concatenationlem} we can get a binary ECC for edit errors with even better parameters (in fact, optimal redundancy) for any $k = O(n^{a})$,  where $a$ can be any constant less than $1$.

\begin{theorem}
For any $n, k \in \mathsf{N}$, there exists an explicit binary error correcting code with message length $n$, codeword length $n+O(k \log n)$ that can correct up to $k$ edit errors. 

\end{theorem}

\begin{proof}
Without loss of generality we can assume that $k \leq n/4$, since for larger $k$ we can just use the asymptotically good code from Theorem \ref{asympGoodECCforInsdel}, where the codeword length is $O(k)$.

Fix a message $x\in \{0,1\}^n$. Let $g$ be the generator from Theorem \ref{gfor3property}. We know that with probability $1-1/\poly(n)$, $ x + g(n, U_r)$ has \ref{uniformproperty1}, \ref{uniformproperty2} and \ref{Bdistinct}. Using an exhaustive search, we can find a seed $u\in \{0,1\}^r$ s.t. $y = x + g(n, u)$ has the three properties. This takes polynomial time since $r = O(\log n)$.

Then we apply the method in Section \ref{sec:randdocexc} to create a redundancy $z$ for $y$, which can correct $3k$ errors, where $|z| = O(k\log n)$.  We further add $u$ to $z$ to form a new sketch $v=z \circ u$ which has size $O(k\log n)$. Now by Theorem~\ref{codeFromSk} we can construct an error correcting code for $y$ with codeword length $n+O(k\log n)$, where the sketch is $v=z \circ u$. 

For decoding, we can first compute $y$, and then compute $x = y + g(n, u)$.
\end{proof}

\section{Discussions and open problems} \label{sec:open} In this paper we constructed deterministic document exchange protocols and binary error correcting codes for edit errors. Our results significantly improve previous results, and in particular we have obtained the first explicit constructions of binary insdel codes that are optimal or almost optimal for a wide range of error parameters $k$. We introduced several new techniques, most notably $\eps$-self matching hash functions and $\eps$-synchronization hash functions. We note that while similar in spirit to $\eps$-self matching strings and $\eps$-synchronization strings introduced in \cite{haeupler2017synchronization}, there are many important differences between these objects. For example, the objects we introduced are hash functions, while the objects in \cite{haeupler2017synchronization} are fixed strings. In particular, although we can also show that random functions are $\eps$-self matching hash functions and $\eps$-synchronization hash functions, the deterministic constructions of $\eps$-self matching hash functions and $\eps$-synchronization hash functions depend on the string $x$. In contrast, $\eps$-self matching strings and $\eps$-synchronization strings do not depend on any input string. 

We believe our techniques can be useful in other applications, and one natural open problem is to get optimal  deterministic document exchange protocols and binary insdel codes for all error parameters.

\bibliographystyle{plain}
\bibliography{ref}

\end{document}